\documentclass[preprint]{sigplanconf}
\usepackage{microtype}
\usepackage[utf8]{inputenc}
\usepackage{newunicodechar}
\usepackage{fancyvrb}
\usepackage{amssymb}
\usepackage{amsmath}
\usepackage{amsthm}
\usepackage{hyperref}
\usepackage{coqdoc}
\usepackage{framed}
\usepackage{proof}
\usepackage{tikz}

\DeclareMathOperator{\CNF}{CNF}
\DeclareMathOperator{\DNF}{DNF}
\DeclareMathOperator{\ENF}{ENF}
\DeclareMathOperator{\Base}{Base}

\DeclareMathOperator{\FV}{FV}

\newcommand{\inl}[1]{\iota_1{#1}}
\newcommand{\inr}[1]{\iota_2{#1}}
\newcommand{\ccase}[5]{\delta({#1},{#2}.{#3},{#4}.{#5})}
\newcommand{\ccaseone}[5]{\delta({#1},{#2}.{#3},\\{#4}.{#5})}

\newcommand{\fst}[1]{\pi_1{#1}}
\newcommand{\snd}[1]{\pi_2{#1}}
\newcommand{\pair}[2]{\langle{#1},{#2}\rangle}

\newcommand{\hsinltwo}[1]{\iota_1{#1}}
\newcommand{\hsinrtwo}[1]{\iota_2{#1}}
\newcommand{\hsinldis}[1]{\iota'_1{#1}}
\newcommand{\hsinrdis}[1]{\iota'_2{#1}}

\newcommand{\hscasn}[3]{\delta({x_{#1}}{#2},{#3})}

\newcommand{\hspair}[2]{\langle{#1},\ldots,{#2}\rangle}
\newcommand{\hspairone}[1]{{#1}}
\newcommand{\hspairfour}[4]{\langle{#1},{#2},{#3},{#4}\rangle}
\newcommand{\hspairtwo}[2]{\langle{#1},{#2}\rangle}

\newcommand{\hstt}{\langle\rangle}

\newcommand{\hsappn}[2]{x_{#1}{#2}}
\newcommand{\hswkn}[1]{\text{w}{#1}}

\DeclareMathOperator{\enfop}{enf}

\newcommand{\enf}[1]{\enfop\left({#1}\right)}

\newcommand{\yslant}{0.5}
\newcommand{\xslant}{-0.9}

\newcommand{\etaeq}{=_\eta}
\newcommand{\etaeqe}{=_\eta^\numbere}
\newcommand{\betaeq}{=_\beta}
\newcommand{\betaeqe}{=_\beta^\numbere}
\newcommand{\betaetaeq}{=_{\beta\eta}}
\newcommand{\betaetaeqe}{=_{\beta\eta}^\numbere}

\newcommand{\binomnb}[2]{\genfrac{}{}{0pt}{}{#1}{#2}}

\newtheorem{theorem}{Theorem}

\theoremstyle{definition}

\newtheorem{convention}{Convention}
\newtheorem{example}{Example}
\theoremstyle{remark}

\newcommand{\ntimes}[2]{{#1}\times{#2}}
\newcommand{\explogn}[2]{{#2}\Rightarrow {#1}}
\newcommand{\explogone}[2]{{#2}\rightrightarrows {#1}}

\newcommand{\numbere}{\mathrm{e}}

\begin{document}

\setlength{\pdfpageheight}{\paperheight}
\setlength{\pdfpagewidth}{\paperwidth}

\conferenceinfo{CONF 'yy}{Month d--d, 20yy, City, ST, Country}
\copyrightyear{20yy}
\copyrightdata{978-1-nnnn-nnnn-n/yy/mm}
\copyrightdoi{nnnnnnn.nnnnnnn}

\publicationrights{licensed}     




\title{The exp-log normal form of types}
\subtitle{Decomposing extensional equality and representing terms compactly}

\authorinfo{Danko Ilik}
{Inria \& LIX, Ecole Polytechnique\\
  91128 Palaiseau Cedex, France}
{danko.ilik@inria.fr}

\maketitle
\begin{abstract} 
  Lambda calculi with algebraic data types lie at the core of
  functional programming languages and proof assistants, but conceal
  at least two fundamental theoretical problems already in the
  presence of the simplest non-trivial data type, the sum type. First,
  we do not know of an explicit and implemented algorithm for deciding
  the beta-eta-equality of terms---and this in spite of the first
  decidability results proven two decades ago. Second, it is not clear
  how to decide when two types are essentially the same,
  i.e. isomorphic, in spite of the meta-theoretic results on
  decidability of the isomorphism.

  In this paper, we present the exp-log normal form of types---derived
  from the representation of exponential polynomials via the unary
  exponential and logarithmic functions---that any
  type built from arrows, products, and sums, can be isomorphically
  mapped to. The type normal form can be used as a simple heuristic
  for deciding type isomorphism, thanks to the fact that it is a
  systematic application of the high-school identities.

  We then show that the type normal form allows to reduce the standard
  beta-eta equational theory of the lambda calculus to a specialized
  version of itself, while preserving the completeness of equality on
  terms. 

  We end by describing an alternative representation of normal terms
  of the lambda calculus with sums, together with a Coq-implemented
  converter into/from our new term calculus. The difference with the
  only other previously implemented heuristic for deciding interesting
  instances of eta-equality by Balat, Di Cosmo, and Fiore, is that we
  exploit the type information of terms substantially and this often
  allows us to obtain a canonical representation of terms without
  performing sophisticated term analyses.
\end{abstract}
\category{Software and its engineering}{Language features}{Abstract data types}
\category{Software and its engineering}{Formal language definitions}{Syntax}
\category{Theory of computation}{Program constructs}{Type structures}


\keywords sum type, eta equality, normal type, normal term, type
isomorphism, type-directed partial evaluation

\section{Introduction}
\label{sec:introduction}

The lambda calculus is a notation for writing functions. Be it
simply-typed or polymorphic, it is also often presented as the core of
modern functional programming languages. Yet, besides functions as
first-class objects, another essential ingredient of these languages
are algebraic data types that typing systems supporting only the
$\to$-type and polymorphism do not model directly. A natural model for
the core of functional languages should at least include direct
support for a simplest case of variant types, \emph{sums}, and
of records i.e. \emph{product} types. But, unlike the theory of the
$\{\to\}$-typed lambda calculus, the theory of the
$\{\to,+,\times\}$-typed one is not all roses.

\paragraph{Canonicity of normal terms and $\eta$-equality} A first
problem is canonicity of normal forms of terms.  Take, for instance,
the term $\lambda x y. y x$ of type
$\tau+\sigma\to (\tau+\sigma\to \rho) \to \rho$, and three of its
\emph{$\eta$-long} representations,
\begin{align*}
  \lambda x. &\lambda y. y \ccase{x}{z}{\inl{z}}{z}{\inr{z}}
  \\
  \lambda x. &\lambda y. \ccase{x}{z}{y (\inl{z})}{z}{y (\inr{z})}
  \\
  \lambda x. &\ccase{x}{z}{\lambda y. y (\inl{z})}{z}{\lambda y. y (\inr{z})},
\end{align*}
where $\delta$ is a pattern matching construct, i.e. a
\texttt{case}-expression analysing the first argument, with branches
of the pattern matching given via the variable $z$ in the second and
third argument.

These three terms are all equal with respect to the standard equational
theory $\betaetaeq$ of the lambda calculus (Figure~\ref{fig:syntax}),
but why should we prefer any one of them over the others to be a
\emph{canonical} representative of the class of equal terms?

Or, consider the following two terms of type
$(\tau_1\to \tau_2)\to(\tau_3\to \tau_1)\to \tau_3\to \tau_4+\tau_5\to
\tau_2$ (example taken from \cite{balat_dicosmo_fiore}):
\begin{gather*}
  \lambda x y z u. x (y z)\\
  \lambda x y z u. \ccase{\ccase{u}{x_1}{\inl{z}}{x_2}{\inr{(y z)}}}{y_1}{x(y y_1)}{y_2}{x y_2}.
\end{gather*}
These terms are $\beta\eta$-equal, but can one easily notice the
equality? In order to do so, since both terms are $\beta$-normal, one
would need to do non-trivial $\beta$- and $\eta$-expansions (see
Example~\ref{ex:2} in Section~\ref{sec:terms}).

For the lambda calculus over the restricted language of types---when the
sum type is absent---these problems do not exist, since
$\beta$-normalization followed by an $\eta$-expansion is deterministic
and produces a canonical representative for any class of
$\beta\eta$-equal terms. Deciding $\betaetaeq$ for that restricted
calculus amounts to comparing canonical forms up to syntactic
identity. 

In the presence of sums, we only have a notion of canonical
\emph{interpretation} of terms in the category of sheaves for the
Grothendieck topology over the category of constrained
environments~\cite{altenkirch_dybjer_hofmann_scott}, as well as the
sophisticated normal form of terms due to Balat, Di Cosmo, and Fiore
which is not canonical (unique) syntactically
\cite{balat_dicosmo_fiore}.  Balat et al. also provide an
implementation of a type-directed partial evaluator that normalizes
terms to their normal form, and this represented up to now the only
implemented \emph{heuristic} for deciding $\beta\eta$-equality---it is
not a full decision procedure, because the normal forms are not
canonical. We shall discuss these and the other decidability results some
more in Related work of Section~\ref{sec:conclusion}.

Treating \emph{full} $\beta\eta$-equality is hard, even if, in
practice, we often only need to treat special cases of it, such as
certain commuting conversions.

\begin{figure*}
  \centering
  \begin{multline*}
    M,N ::= x^\tau ~|~ (M^{\tau\to\sigma} N^\tau)^\sigma ~|~ (\fst{M^{\tau\times\sigma}})^\tau ~|~ (\snd{M^{\tau\times\sigma}})^\sigma ~|~ \ccase{M^{\tau+\sigma}}{x_1^\tau}{N_1^\rho}{x_2^\sigma}{N_2^\rho}^\rho ~|~ \\
     | (\lambda x^\tau.M^\sigma)^{\tau\to\sigma} ~|~ \pair{M^\tau}{N^\sigma}^{\tau\times\sigma} ~|~ (\inl{M^\tau})^{\tau+\sigma} ~|~ (\inr{M^\sigma})^{\tau+\sigma}
  \end{multline*}
  \begin{align}
    \label{beta:arrow}\tag{$\beta_\to$}(\lambda x. N) M &\betaeq N\{M/x\} & \\
    \label{beta:product}\tag{$\beta_\times$}\pi_i{\pair{M_1}{M_2}} &\betaeq M_i\\
    \label{beta:sum}\tag{$\beta_+$}\ccase{\iota_i{M}}{x_1}{N_1}{x_2}{N_2} &\betaeq N_i\{M/x_i\} & \\
    \label{eta:arrow}\tag{$\eta_\to$}N &\etaeq \lambda x. N x & x\not\in\FV(N)\\
    \label{eta:product}\tag{$\eta_\times$}N &\etaeq \pair{\fst{N}}{\snd{N}}\\
    \label{eta:sum}\tag{$\eta_+$}N\{M/x\}&\etaeq \ccase{M}{x_1}{N\{\inl{x_1}/x\}}{x_2}{N\{\inr{x_2}/x\}} & x_1,x_2\not\in\FV(N)
  \end{align}
  \caption{Terms of the $\{\to,+,\times\}$-typed lambda calculus and
    axioms of the equational theory $\betaetaeq$ between typed terms.}
  \label{fig:syntax}
\end{figure*}

\paragraph{Recognizing isomorphic types}
If we leave aside the problems of canonicity of and equality between
terms, there is a further problem at the level of \emph{types} that
makes it hard to determine whether two type signatures are essentially
the same one. Namely, although for each of the type languages
$\{\to,\times\}$ and $\{\to,+\}$ there is a very simple algorithm for
deciding \emph{type isomorphism}, for the whole of the language
$\{\to,+,\times\}$ it is only known that type isomorphism is decidable
when types are to be interpreted as finite structures, and that
without a practically implementable algorithm in
sight~\cite{ilik_sum_axioms}.

The importance of deciding type isomorphism for functional programming
has been recognized early on by Rittri~\cite{Rittri}, who proposed to
use it as a criterium for searching over a library of functional
subroutines. Two types being isomorphic means that one can switch
programs and data back and forth between the types without loss of
information. Recently, type isomorphisms have also become popular in
the community around homotopy type theory.

It is embarrassing that there are no algorithms for deciding type
isomorphism for such an ubiquitous type system. Finally, even if
finding an implementable decision procedure for the \emph{full} type
language $\{\to,+,\times\}$ were hard, might we simply be able to
cover fragments that are important in practice?

\paragraph{Organization of this paper}

In this paper, we shall be treating the two kinds of problems
explained above simultaneously, not as completely distinct ones:
traditionally, studies of canonical forms and deciding equality on
terms have used very little of the type information annotating the
terms (with the exceptions mentioned in the concluding
Section~\ref{sec:conclusion}).

We shall start by introducing in Section~\ref{sec:types} a normal form
for \emph{types}---called the \emph{exp-log normal form} (ENF)---that
preserves the isomorphism between the source and the target type; we
shall also give an implementation, a purely functional one, that can
be used as a heuristic procedure for deciding isomorphism of two
types.

Even if reducing a type to its ENF does not present a complete
decision procedure for isomorphism of \emph{types}, we shall show in
the subsequent Section~\ref{sec:equality} that it has dramatic effects
on the theory of $\beta\eta$-equality of \emph{terms}. Namely, one can
reduce the problem of showing equality for the standard $\betaetaeq$
relation to the problem of showing it for a new equality theory
$\betaetaeqe$ (Figure~\ref{fig:syntax:enf})---this later being a
\emph{specialization} of $\betaetaeq$. That is, a complete
axiomatization of $\beta\eta$-equality that is a strict subset of the
currently standard one is possible.

In Section~\ref{sec:terms}, we shall go further and describe a
minimalist calculus of terms---\emph{compact terms} at ENF type---that
can be used as an alternative to the usual lambda calculus with
sums. With its properties of a syntactic simplification of the later
(for instance, there is no lambda abstraction), the new calculus
allows a more canonical representation of terms. We show that, for a
number of interesting examples, converting lambda terms to compact
terms and comparing the obtained terms for syntactic identity provides
a simple heuristic for deciding $\betaetaeq$.

The paper is accompanied by a prototype normalizing converter between
lambda- and compact terms implemented in Coq.


\section{The exp-log normal form of types}
\label{sec:types}

The trouble with sums starts already at the level of types.  Namely,
when we consider types built from function spaces, products, and
disjoint unions (sums),
\[
\tau,\sigma ::= \chi_i ~|~ \tau \to \sigma ~|~ \tau\times\sigma ~|~ \tau+\sigma,
\]
where $\chi_i$ are atomic types (or type variables), it is not always
clear when two given types are essentially the same one. More
precisely, it is not known \emph{how} to decide whether two types are
isomorphic \cite{ilik_sum_axioms}. Although the notion of isomorphism
can be treated abstractly in Category Theory, in bi-Cartesian closed
categories, and without committing to a specific term calculus
inhabiting the types, in the language of the standard syntax and
equational theory of lambda calculus with sums
(Figure~\ref{fig:syntax}), the types $\tau$ and $\sigma$ are
isomorphic when there exist coercing lambda terms $M : \sigma\to\tau$
and $N : \tau\to\sigma$ such that
\[
\lambda x. M (N x) =_{\beta\eta} \lambda x. x
\quad \text{ and } \quad
\lambda y. N (M y) =_{\beta\eta} \lambda y. y.
\]
In other words, data/programs can be converted back and forth between
$\tau$ and $\sigma$ without loss of information.

The problem of isomorphism is in fact closely related 
to the famous Tarski High School Identities Problem
\cite{burris04,fiore06}. What is important for us here is that
\emph{types can be seen as just arithmetic expressions}: if the type
$\tau\to\sigma$ is denoted by the binary arithmetic exponentiation
$\sigma^\tau$, then every type $\rho$ denotes at the same time an
\emph{exponential} polynomial $\rho$. The difference with ordinary
polynomials is that the exponent can now also contain a (type)
variable, while exponentiation in ordinary polynomials is always of
the form $\sigma^n$ for a concrete $n\in\mathbb{N}$ i.e.
$\sigma^n =
\underbrace{\sigma\times\cdots\times\sigma}_{n\text{-times}}$.
Moreover, we have that
\[
\tau\cong\sigma \text{ implies } \mathbb{N}^+ \vDash \tau=\sigma,
\]
that is, type isomorphism implies that arithmetic equality holds for
any substitution of variables by positive natural numbers.

This hence provides an procedure for proving \emph{non}-isomorphism:
given two types, prove they are not equal as exponential polynomials,
and that means they cannot possibly be isomorphic. But, we are
interested in a positive decision procedure. Such a procedure exists
for both the languages of types $\{\to,\times\}$ and $\{\times,+\}$,
since then we have an equivalence:
\[
\tau\cong\sigma \text{ iff } \mathbb{N}^+ \vDash \tau=\sigma.
\]
Indeed, in these cases type isomorphism can not only be decided, but
also effectively built. In the case of $\{\times,+\}$, the procedure
amounts to transforming the type to disjunctive normal form, or the
(\emph{non}-exponential) polynomial to canonical form, while in that
of $\{\to,\times\}$, there is a canonical normal form obtained by type
transformation that follows currying \cite{Rittri}.

Given that it is not known whether one can find such a canonical
normal form for the full language of types \cite{ilik_sum_axioms},
what we can hope to do in practice is to find at least a
\emph{pseudo}-canonical normal form. We shall now define such a type
normal form.

The idea is to use the decomposition of the binary exponential
function $\sigma^\tau$ through unary exponentiation and
logarithm. This is a well known transformation in Analysis, where for
the natural logarithm and Euler's number $\numbere$ we would use
\[
\sigma^\tau = \numbere^{\tau\times\log\sigma}\quad\text{also written}\quad\sigma^\tau = \exp(\tau\times\log\sigma).
\]
The systematic study of such normal forms by Du Bois-Reymond described
in the book \cite{hardy} served us as inspiration.

But how exactly are we to go about using this equality for types when
it uses logarithms i.e. transcendental numbers? Luckily, we do not
have to think of real numbers at all, because what is described above
can be seen through the eyes of abstract Algebra, in exponential
fields, as a pair of mutually inverse homomorphisms $\exp$ and $\log$
between the multiplicative and additive group, satisfying
\begin{align*}
  \exp(\tau_1+\tau_2) &= \exp\tau_1\times \exp\tau_2 & \exp(\log\tau) &= \tau\\
  \log(\tau_1\times\tau_2) &= \log\tau_1+\log\tau_2 & \log(\exp\tau) &= \tau.
\end{align*}
In other words, $\exp$ and $\log$ can be considered as macro
expansions rather than unary type constructors.  Let us take the type
$\tau+\sigma\to (\tau+\sigma\to \rho) \to \rho$ from
Section~\ref{sec:introduction}, assuming for simplicity that
$\tau,\sigma,\rho$ are atomic types. It can be normalized in the
following way:
\begin{multline*}
    \tau+\sigma\to (\tau+\sigma\to \rho) \to \rho =\\
  = \left(\rho^{\rho^{\tau+\sigma}}\right)^{\tau+\sigma} =\\
  = \exp((\tau+\sigma)\log[\exp\{\exp((\tau+\sigma)\log\rho)\log\rho\}]) \rightsquigarrow\\
  \rightsquigarrow \exp((\tau+\sigma)\log[\exp\{\exp(\tau\log\rho)\exp(\sigma\log\rho)\log\rho\}]) \rightsquigarrow\\
  \rightsquigarrow \exp((\tau+\sigma)\exp(\tau\log\rho)\exp(\sigma\log\rho)\log\rho) \rightsquigarrow\\
  \rightsquigarrow \exp\big(\tau\exp(\tau\log\rho)\exp(\sigma\log\rho)\log\rho)\\\exp(\sigma\exp(\tau\log\rho)\exp(\sigma\log\rho)\log\rho\big) = \\
  = \rho^{\tau\rho^\tau\rho^\sigma}\rho^{\sigma\rho^\tau\rho^\sigma} \\
  = (\tau\times(\tau\to\rho)\times(\sigma\to\rho)\to\rho)\times(\sigma\times(\tau\to\rho)\times(\sigma\to\rho)\to\rho).
\end{multline*}

As the exp-log transformation of arrow types is at the source of this
type normalization procedure, we call the obtained normal form
\emph{the exp-log normal form (ENF)}. Be believe the link to abstract
algebra is well work keeping in mind, since it may give rise to
further cross-fertilization between mathematics and the theory of
programming languages. However, from the operational point of view,
all this transformation does is that it \emph{prioritized} and
\emph{orients} the high-school identities,
\begin{align}
  \label{eq:hsi:assoc:plus}(f+g)+h &\rightsquigarrow f+(g+h) 
  \\
  \label{eq:hsi:assoc:times}(f g) h &\rightsquigarrow f(g h) 
  \\
  \label{eq:hsi:distrib:1}f(g+h) &\rightsquigarrow f g + f h 
  \\
  \label{eq:hsi:distrib:2}(f+g) h &\rightsquigarrow f h + g h 
  \\
  \label{eq:hsi:exp:1}f^{g+h} &\rightsquigarrow f^g f^h 
  \\
  \label{eq:hsi:exp:2}(f g)^h &\rightsquigarrow f^h g^h 
  \\
  \label{eq:hsi:exp:3}(f^g)^h &\rightsquigarrow f^{h g} 
                                                                     ,
\end{align}
all of which are valid as type isomorphisms. We can thus also compute
the \emph{isomorphic} normal form of the type directly, for instance
for the second example of Section~\ref{sec:introduction}:
\begin{multline*}
  (\tau_1\to \tau_2)\to(\tau_3\to \tau_1)\to \tau_3\to \tau_4+\tau_5\to \tau_2 =\\
  = \left(\left(\left(\tau_2^{\tau_4+\tau_5}\right)^{\tau_3}\right)^{{\tau_1}^{\tau_3}}\right)^{{\tau_2}^{\tau_1}} \rightsquigarrow\\
  \rightsquigarrow \tau_2^{{\tau_2^{\tau_1}}{\tau_1^{\tau_3}}\tau_3\tau_4}\tau_2^{{\tau_2^{\tau_1}}{\tau_1^{\tau_3}}\tau_3\tau_5} =\\
  = \big(\tau_4\times\tau_3\times(\tau_3\to\tau_1)\times(\tau_1\to\tau_2)\to\tau_2\big)\times\\\big(\tau_5\times\tau_3\times(\tau_3\to\tau_1)\times(\tau_1\to\tau_2)\to\tau_2\big).
\end{multline*}
Of course, some care needs to be taken when applying the rewrite
rules, in order for the procedure to be deterministic, like giving
precedence to the type rewrite rules and normalizing
sub-expressions. To be precise, we provide a purely functional Coq
implementation below.  This is just one possible implementation of the
rewriting rules, but being purely functional and structurally
recursive (i.e. terminating) it allows us to understand the
restrictions imposed on types in normal form, as it proves the
following theorem.

\begin{theorem}\label{thm:inductive} If $\tau$ is a type in exp-log normal form, then
  $\tau\in\ENF$, where
  \begin{align*}
    \ENF\ni e &::= c ~|~ d,
  \end{align*}
  where
  \begin{align*}
    \DNF\ni d,d_i &::= c_1 + (c_2 + (\cdots + n)\cdots) & n\ge 2\\
    \CNF\ni c,c_i &::= (c_1\to b_1)\times(\cdots\times(c_n\to b_n)\cdots) & n\ge 0\\
    \Base\ni b,b_i &::= p ~|~ d,
  \end{align*}
  and $p$ denotes atomic types (type variables).
\end{theorem}

Assuming a given set of atomic types,
\begin{framed}
\begin{coqdoccode}
\coqdocnoindent
\coqdockw{Parameter} \coqdocvar{Proposition} : \coqdockw{Set}.\coqdoceol
\end{coqdoccode}
\end{framed}
\noindent
the goal is to map the unrestricted language of types, given by the
inductive definition,\footnote{May the reader to forgive us for the
  implicit use of the Curry-Howard correspondence in the Coq code
  snippets, where we refer to types and type constructors as formulas
  and formula constructors.}
\begin{framed}
\begin{coqdoccode}
\coqdocnoindent
\coqdockw{Inductive} \coqdocvar{Formula} : \coqdockw{Set} :=\coqdoceol
\coqdocnoindent
\ensuremath{|} \coqdocvar{prop} : \coqdocvar{Proposition} \ensuremath{\rightarrow} \coqdocvar{Formula}\coqdoceol
\coqdocnoindent
\ensuremath{|} \coqdocvar{disj} : \coqdocvar{Formula} \ensuremath{\rightarrow} \coqdocvar{Formula} \ensuremath{\rightarrow} \coqdocvar{Formula}\coqdoceol
\coqdocnoindent
\ensuremath{|} \coqdocvar{conj} : \coqdocvar{Formula} \ensuremath{\rightarrow} \coqdocvar{Formula} \ensuremath{\rightarrow} \coqdocvar{Formula}\coqdoceol
\coqdocnoindent
\ensuremath{|} \coqdocvar{impl} : \coqdocvar{Formula} \ensuremath{\rightarrow} \coqdocvar{Formula} \ensuremath{\rightarrow} \coqdocvar{Formula}.\coqdoceol
\end{coqdoccode}
  
\end{framed}
\noindent
into the exp-log normal form which fits in the following inductive signature.
\begin{framed}
\begin{coqdoccode}
\coqdocnoindent
\coqdockw{Inductive} \coqdocvar{CNF} : \coqdockw{Set} :=\coqdoceol
\coqdocnoindent
\ensuremath{|} \coqdocvar{top}\coqdoceol
\coqdocnoindent
\ensuremath{|} \coqdocvar{con} : \coqdocvar{CNF} \ensuremath{\rightarrow} \coqdocvar{Base} \ensuremath{\rightarrow} \coqdocvar{CNF} \ensuremath{\rightarrow} \coqdocvar{CNF}\coqdoceol
\coqdocnoindent
\coqdockw{with} \coqdocvar{DNF} : \coqdockw{Set} :=\coqdoceol
\coqdocnoindent
\ensuremath{|} \coqdocvar{two} : \coqdocvar{CNF} \ensuremath{\rightarrow} \coqdocvar{CNF} \ensuremath{\rightarrow} \coqdocvar{DNF}\coqdoceol
\coqdocnoindent
\ensuremath{|} \coqdocvar{dis} : \coqdocvar{CNF} \ensuremath{\rightarrow} \coqdocvar{DNF} \ensuremath{\rightarrow} \coqdocvar{DNF}\coqdoceol
\coqdocnoindent
\coqdockw{with} \coqdocvar{Base} : \coqdockw{Set} :=\coqdoceol
\coqdocnoindent
\ensuremath{|} \coqdocvar{prp} : \coqdocvar{Proposition} \ensuremath{\rightarrow} \coqdocvar{Base}\coqdoceol
\coqdocnoindent
\ensuremath{|} \coqdocvar{bd} : \coqdocvar{DNF} \ensuremath{\rightarrow} \coqdocvar{Base}.\coqdoceol
\coqdocemptyline
\coqdocnoindent
\coqdockw{Inductive} \coqdocvar{ENF} : \coqdockw{Set} :=\coqdoceol
\coqdocnoindent
\ensuremath{|} \coqdocvar{cnf} : \coqdocvar{CNF} \ensuremath{\rightarrow} \coqdocvar{ENF}\coqdoceol
\coqdocnoindent
\ensuremath{|} \coqdocvar{dnf} : \coqdocvar{DNF} \ensuremath{\rightarrow} \coqdocvar{ENF}.\coqdoceol
\end{coqdoccode}
  
\end{framed}
\noindent
The $\mathit{con}~c_1~b~c_2$ constructor corresponds to $b^{c_1} c_2$
or the type $(c_1\to b)\times c_2$ from Theorem~\ref{thm:inductive}.
The normalization function, $\enf{\cdot}$,
\begin{framed}
\begin{coqdoccode}
\coqdocindent{0.00em}
\coqdockw{Fixpoint} \coqdocvar{enf} (\coqdocvar{f} : \coqdocvar{Formula}) \{\coqdockw{struct} \coqdocvar{f}\} : \coqdocvar{ENF} :=\coqdoceol
\coqdocindent{1.00em}
\coqdockw{match} \coqdocvar{f} \coqdockw{with}\coqdoceol
\coqdocindent{1.00em}
\ensuremath{|} \coqdocvar{prop} \coqdocvar{p} \ensuremath{\Rightarrow} \coqdocvar{cnf} (\coqdocvar{p2c} \coqdocvar{p})\coqdoceol
\coqdocindent{1.00em}
\ensuremath{|} \coqdocvar{disj} \coqdocvar{f0} \coqdocvar{f1} \ensuremath{\Rightarrow} \coqdocvar{dnf} (\coqdocvar{nplus} (\coqdocvar{enf} \coqdocvar{f0}) (\coqdocvar{enf} \coqdocvar{f1}))\coqdoceol
\coqdocindent{1.00em}
\ensuremath{|} \coqdocvar{conj} \coqdocvar{f0} \coqdocvar{f1} \ensuremath{\Rightarrow} \coqdocvar{distrib} (\coqdocvar{enf} \coqdocvar{f0}) (\coqdocvar{enf} \coqdocvar{f1})\coqdoceol
\coqdocindent{1.00em}
\ensuremath{|} \coqdocvar{impl} \coqdocvar{f0} \coqdocvar{f1} \ensuremath{\Rightarrow} \coqdocvar{cnf} (\coqdocvar{explogn} (\coqdocvar{enf2cnf} (\coqdocvar{enf} \coqdocvar{f1})) (\coqdocvar{enf} \coqdocvar{f0}))\coqdoceol
\coqdocindent{1.00em}
\coqdockw{end}.\coqdoceol
\end{coqdoccode}

\end{framed}
\noindent
is defined using the following fixpoints:
\begin{description}
\item[nplus] which makes a flattened $n$-ary sum out of two given
  $n$-ary sums, i.e. implements the $+$-associativity rewriting
  (\ref{eq:hsi:assoc:plus}),
\item[ntimes] which is analogous to `nplus', but for products,
  implementing (\ref{eq:hsi:assoc:times}),
\item[distrib] which performs the distributivity rewriting,
  (\ref{eq:hsi:distrib:1}) and (\ref{eq:hsi:distrib:2}), and
\item[explogn] which performs the rewriting involving exponentiations,
  (\ref{eq:hsi:exp:1}), (\ref{eq:hsi:exp:2}), and
  (\ref{eq:hsi:exp:3}).
\end{description}
\begin{framed}
\begin{coqdoccode}
\coqdocindent{0.00em}
\coqdockw{Fixpoint} \coqdocvar{nplus1} (\coqdocvar{d} : \coqdocvar{DNF})(\coqdocvar{e2} : \coqdocvar{ENF}) \{\coqdockw{struct} \coqdocvar{d}\} : \coqdocvar{DNF} :=\coqdoceol
\coqdocindent{1.00em}
\coqdockw{match} \coqdocvar{d} \coqdockw{with}\coqdoceol
\coqdocindent{1.00em}
\ensuremath{|} \coqdocvar{two} \coqdocvar{c} \coqdocvar{c0} \ensuremath{\Rightarrow} \coqdockw{match} \coqdocvar{e2} \coqdockw{with}\coqdoceol
\coqdocindent{8.00em}
\ensuremath{|} \coqdocvar{cnf} \coqdocvar{c1} \ensuremath{\Rightarrow} \coqdocvar{dis} \coqdocvar{c} (\coqdocvar{two} \coqdocvar{c0} \coqdocvar{c1})\coqdoceol
\coqdocindent{8.00em}
\ensuremath{|} \coqdocvar{dnf} \coqdocvar{d0} \ensuremath{\Rightarrow} \coqdocvar{dis} \coqdocvar{c} (\coqdocvar{dis} \coqdocvar{c0} \coqdocvar{d0})\coqdoceol
\coqdocindent{8.00em}
\coqdockw{end}\coqdoceol
\coqdocindent{1.00em}
\ensuremath{|} \coqdocvar{dis} \coqdocvar{c} \coqdocvar{d0} \ensuremath{\Rightarrow} \coqdocvar{dis} \coqdocvar{c} (\coqdocvar{nplus1} \coqdocvar{d0} \coqdocvar{e2})\coqdoceol
\coqdocindent{1.00em}
\coqdockw{end}.\coqdoceol
\coqdocemptyline
\coqdocindent{0.00em}
\coqdockw{Definition} \coqdocvar{nplus} (\coqdocvar{e1} \coqdocvar{e2} : \coqdocvar{ENF}) : \coqdocvar{DNF} :=\coqdoceol
\coqdocindent{1.00em}
\coqdockw{match} \coqdocvar{e1} \coqdockw{with}\coqdoceol
\coqdocindent{1.00em}
\ensuremath{|} \coqdocvar{cnf} \coqdocvar{a} \ensuremath{\Rightarrow} \coqdockw{match} \coqdocvar{e2} \coqdockw{with}\coqdoceol
\coqdocindent{6.50em}
\ensuremath{|} \coqdocvar{cnf} \coqdocvar{c} \ensuremath{\Rightarrow} \coqdocvar{two} \coqdocvar{a} \coqdocvar{c}\coqdoceol
\coqdocindent{6.50em}
\ensuremath{|} \coqdocvar{dnf} \coqdocvar{d} \ensuremath{\Rightarrow} \coqdocvar{dis} \coqdocvar{a} \coqdocvar{d}\coqdoceol
\coqdocindent{6.50em}
\coqdockw{end}\coqdoceol
\coqdocindent{1.00em}
\ensuremath{|} \coqdocvar{dnf} \coqdocvar{b} \ensuremath{\Rightarrow} \coqdocvar{nplus1} \coqdocvar{b} \coqdocvar{e2}\coqdoceol
\coqdocindent{1.00em}
\coqdockw{end}.\coqdoceol
\coqdocemptyline
\coqdocindent{0.00em}
\coqdockw{Fixpoint} \coqdocvar{ntimes} (\coqdocvar{c1} \coqdocvar{c2} : \coqdocvar{CNF}) \{\coqdockw{struct} \coqdocvar{c1}\} : \coqdocvar{CNF} :=\coqdoceol
\coqdocindent{1.00em}
\coqdockw{match} \coqdocvar{c1} \coqdockw{with}\coqdoceol
\coqdocindent{1.00em}
\ensuremath{|} \coqdocvar{top} \ensuremath{\Rightarrow} \coqdocvar{c2}\coqdoceol
\coqdocindent{1.00em}
\ensuremath{|} \coqdocvar{con} \coqdocvar{c10} \coqdocvar{d} \coqdocvar{c13} \ensuremath{\Rightarrow} \coqdocvar{con} \coqdocvar{c10} \coqdocvar{d} (\coqdocvar{ntimes} \coqdocvar{c13} \coqdocvar{c2})\coqdoceol
\coqdocindent{1.00em}
\coqdockw{end}.\coqdoceol
\coqdocemptyline
\coqdocindent{0.00em}
\coqdockw{Fixpoint} \coqdocvar{distrib0} (\coqdocvar{c} : \coqdocvar{CNF})(\coqdocvar{d} : \coqdocvar{DNF}) : \coqdocvar{ENF} :=\coqdoceol
\coqdocindent{1.00em}
\coqdockw{match} \coqdocvar{d} \coqdockw{with}\coqdoceol
\coqdocindent{1.00em}
\ensuremath{|} \coqdocvar{two} \coqdocvar{c0} \coqdocvar{c1} \ensuremath{\Rightarrow} \coqdocvar{dnf} (\coqdocvar{two} (\coqdocvar{ntimes} \coqdocvar{c} \coqdocvar{c0}) (\coqdocvar{ntimes} \coqdocvar{c} \coqdocvar{c1}))\coqdoceol
\coqdocindent{1.00em}
\ensuremath{|} \coqdocvar{dis} \coqdocvar{c0} \coqdocvar{d0} \ensuremath{\Rightarrow} \coqdocvar{dnf} \coqdockw{match} \coqdocvar{distrib0} \coqdocvar{c} \coqdocvar{d0} \coqdockw{with}\coqdoceol
\coqdocindent{10.50em}
\ensuremath{|} \coqdocvar{cnf} \coqdocvar{c1} \ensuremath{\Rightarrow} \coqdocvar{two} (\coqdocvar{ntimes} \coqdocvar{c} \coqdocvar{c0}) \coqdocvar{c1}\coqdoceol
\coqdocindent{10.50em}
\ensuremath{|} \coqdocvar{dnf} \coqdocvar{d1} \ensuremath{\Rightarrow} \coqdocvar{dis} (\coqdocvar{ntimes} \coqdocvar{c} \coqdocvar{c0}) \coqdocvar{d1}\coqdoceol
\coqdocindent{10.50em}
\coqdockw{end}\coqdoceol
\coqdocindent{1.00em}
\coqdockw{end}.\coqdoceol
\coqdocemptyline
\coqdocindent{0.00em}
\coqdockw{Definition} \coqdocvar{distrib1} (\coqdocvar{c} : \coqdocvar{CNF})(\coqdocvar{e} : \coqdocvar{ENF}) : \coqdocvar{ENF} :=\coqdoceol
\coqdocindent{1.00em}
\coqdockw{match} \coqdocvar{e} \coqdockw{with}\coqdoceol
\coqdocindent{1.00em}
\ensuremath{|} \coqdocvar{cnf} \coqdocvar{a} \ensuremath{\Rightarrow} \coqdocvar{cnf} (\coqdocvar{ntimes} \coqdocvar{c} \coqdocvar{a})\coqdoceol
\coqdocindent{1.00em}
\ensuremath{|} \coqdocvar{dnf} \coqdocvar{b} \ensuremath{\Rightarrow} \coqdocvar{distrib0} \coqdocvar{c} \coqdocvar{b}\coqdoceol
\coqdocindent{1.00em}
\coqdockw{end}.\coqdoceol
\coqdocemptyline
\coqdocindent{0.00em}
\coqdockw{Fixpoint} \coqdocvar{explog0} (\coqdocvar{d} : \coqdocvar{Base})(\coqdocvar{d2} : \coqdocvar{DNF}) \{\coqdockw{struct} \coqdocvar{d2}\} : \coqdocvar{CNF} :=\coqdoceol
\coqdocindent{1.00em}
\coqdockw{match} \coqdocvar{d2} \coqdockw{with}\coqdoceol
\coqdocindent{1.00em}
\ensuremath{|} \coqdocvar{two} \coqdocvar{c1} \coqdocvar{c2} \ensuremath{\Rightarrow} \coqdocvar{ntimes} (\coqdocvar{con} \coqdocvar{c1} \coqdocvar{d} \coqdocvar{top}) (\coqdocvar{con} \coqdocvar{c2} \coqdocvar{d} \coqdocvar{top})\coqdoceol
\coqdocindent{1.00em}
\ensuremath{|} \coqdocvar{dis} \coqdocvar{c} \coqdocvar{d3} \ensuremath{\Rightarrow} \coqdocvar{ntimes} (\coqdocvar{con} \coqdocvar{c} \coqdocvar{d} \coqdocvar{top}) (\coqdocvar{explog0} \coqdocvar{d} \coqdocvar{d3})\coqdoceol
\coqdocindent{1.00em}
\coqdockw{end}.\coqdoceol
\coqdocemptyline
\coqdocindent{0.00em}
\coqdockw{Definition} \coqdocvar{explog1} (\coqdocvar{d} : \coqdocvar{Base})(\coqdocvar{e} : \coqdocvar{ENF}) : \coqdocvar{CNF} :=\coqdoceol
\coqdocindent{1.00em}
\coqdockw{match} \coqdocvar{e} \coqdockw{with}\coqdoceol
\coqdocindent{1.00em}
\ensuremath{|} \coqdocvar{cnf} \coqdocvar{c} \ensuremath{\Rightarrow} \coqdocvar{con} \coqdocvar{c} \coqdocvar{d} \coqdocvar{top}\coqdoceol
\coqdocindent{1.00em}
\ensuremath{|} \coqdocvar{dnf} \coqdocvar{d1} \ensuremath{\Rightarrow} \coqdocvar{explog0} \coqdocvar{d} \coqdocvar{d1}\coqdoceol
\coqdocindent{1.00em}
\coqdockw{end}.\coqdoceol
\coqdocemptyline
\coqdocindent{0.00em}
\coqdockw{Fixpoint} \coqdocvar{distribn} (\coqdocvar{d} : \coqdocvar{DNF})(\coqdocvar{e2} : \coqdocvar{ENF}) \{\coqdockw{struct} \coqdocvar{d}\} : \coqdocvar{ENF} :=\coqdoceol
\coqdocindent{1.00em}
\coqdockw{match} \coqdocvar{d} \coqdockw{with}\coqdoceol
\coqdocindent{1.00em}
\ensuremath{|} \coqdocvar{two} \coqdocvar{c} \coqdocvar{c0} \ensuremath{\Rightarrow} \coqdocvar{dnf} (\coqdocvar{nplus} (\coqdocvar{distrib1} \coqdocvar{c} \coqdocvar{e2}) (\coqdocvar{distrib1} \coqdocvar{c0} \coqdocvar{e2}))\coqdoceol
\coqdocindent{1.00em}
\ensuremath{|} \coqdocvar{dis} \coqdocvar{c} \coqdocvar{d0} \ensuremath{\Rightarrow} \coqdocvar{dnf} (\coqdocvar{nplus} (\coqdocvar{distrib1} \coqdocvar{c} \coqdocvar{e2}) (\coqdocvar{distribn} \coqdocvar{d0} \coqdocvar{e2}))\coqdoceol
\coqdocindent{1.00em}
\coqdockw{end}.\coqdoceol
\coqdocemptyline
\coqdocindent{0.00em}
\coqdockw{Definition} \coqdocvar{distrib} (\coqdocvar{e1} \coqdocvar{e2} : \coqdocvar{ENF}) : \coqdocvar{ENF} :=\coqdoceol
\coqdocindent{1.00em}
\coqdockw{match} \coqdocvar{e1} \coqdockw{with}\coqdoceol
\coqdocindent{1.00em}
\ensuremath{|} \coqdocvar{cnf} \coqdocvar{a} \ensuremath{\Rightarrow} \coqdocvar{distrib1} \coqdocvar{a} \coqdocvar{e2}\coqdoceol
\coqdocindent{1.00em}
\ensuremath{|} \coqdocvar{dnf} \coqdocvar{b} \ensuremath{\Rightarrow} \coqdocvar{distribn} \coqdocvar{b} \coqdocvar{e2}\coqdoceol
\coqdocindent{1.00em}
\coqdockw{end}.\coqdoceol
\coqdocemptyline
\coqdocindent{0.00em}
\coqdockw{Fixpoint} \coqdocvar{explogn} (\coqdocvar{c}:\coqdocvar{CNF})(\coqdocvar{e2}:\coqdocvar{ENF}) \{\coqdockw{struct} \coqdocvar{c}\} : \coqdocvar{CNF} :=\coqdoceol
\coqdocindent{1.00em}
\coqdockw{match} \coqdocvar{c} \coqdockw{with}\coqdoceol
\coqdocindent{1.00em}
\ensuremath{|} \coqdocvar{top} \ensuremath{\Rightarrow} \coqdocvar{top}\coqdoceol
\coqdocindent{1.00em}
\ensuremath{|} \coqdocvar{con} \coqdocvar{c1} \coqdocvar{d} \coqdocvar{c2} \ensuremath{\Rightarrow}\coqdoceol
\coqdocindent{2.00em}
\coqdocvar{ntimes} (\coqdocvar{explog1} \coqdocvar{d} (\coqdocvar{distrib1} \coqdocvar{c1} \coqdocvar{e2})) (\coqdocvar{explogn} \coqdocvar{c2} \coqdocvar{e2})\coqdoceol
\coqdocindent{1.00em}
\coqdockw{end}.\coqdoceol
\coqdocemptyline
\coqdocindent{0.00em}
\coqdockw{Definition} \coqdocvar{p2c} : \coqdocvar{Proposition} \ensuremath{\rightarrow} \coqdocvar{CNF} :=\coqdoceol
\coqdocindent{1.00em}
\coqdockw{fun} \coqdocvar{p} \ensuremath{\Rightarrow} \coqdocvar{con} \coqdocvar{top} (\coqdocvar{prp} \coqdocvar{p}) \coqdocvar{top}.\coqdoceol
\coqdocemptyline
\coqdocindent{0.00em}
\coqdockw{Definition} \coqdocvar{b2c} : \coqdocvar{Base} \ensuremath{\rightarrow} \coqdocvar{CNF} :=\coqdoceol
\coqdocindent{1.00em}
\coqdockw{fun} \coqdocvar{b} \ensuremath{\Rightarrow}\coqdoceol
\coqdocindent{2.00em}
\coqdockw{match} \coqdocvar{b} \coqdockw{with}\coqdoceol
\coqdocindent{2.00em}
\ensuremath{|} \coqdocvar{prp} \coqdocvar{p} \ensuremath{\Rightarrow} \coqdocvar{p2c} \coqdocvar{p}\coqdoceol
\coqdocindent{2.00em}
\ensuremath{|} \coqdocvar{bd} \coqdocvar{d} \ensuremath{\Rightarrow} \coqdocvar{con} \coqdocvar{top} (\coqdocvar{bd} \coqdocvar{d}) \coqdocvar{top}\coqdoceol
\coqdocindent{2.00em}
\coqdockw{end}.\coqdoceol
\coqdocemptyline
\coqdocindent{0.00em}
\coqdockw{Fixpoint} \coqdocvar{enf2cnf} (\coqdocvar{e}:\coqdocvar{ENF}) \{\coqdockw{struct} \coqdocvar{e}\} : \coqdocvar{CNF} :=\coqdoceol
\coqdocindent{1.00em}
\coqdockw{match} \coqdocvar{e} \coqdockw{with}\coqdoceol
\coqdocindent{1.00em}
\ensuremath{|} \coqdocvar{cnf} \coqdocvar{c} \ensuremath{\Rightarrow} \coqdocvar{c}\coqdoceol
\coqdocindent{1.00em}
\ensuremath{|} \coqdocvar{dnf} \coqdocvar{d} \ensuremath{\Rightarrow} \coqdocvar{b2c} (\coqdocvar{bd} \coqdocvar{d})\coqdoceol
\coqdocindent{1.00em}
\coqdockw{end}.\coqdoceol
\end{coqdoccode}
  
\end{framed}

From the inductive characterization of the previous theorem, it is
immediate to notice that the exp-log normal form (ENF) is in fact a
combination of disjunctive- (DNF) and conjunctive normal forms (CNF),
and their extension to also cover the function type. We shall now
apply this simple and loss-less transformation of types to the
equational theory of terms of the lambda calculus with sums.


\section{$\beta\eta$-Congruence classes at ENF type}
\label{sec:equality}

The virtue of type isomorphisms is that they preserve the equational
theory of the term calculus: an isomorphism between $\tau$ and
$\sigma$ is witnessed by a pair of lambda terms
\[T : \sigma\to\tau
  \quad\text{and}\quad 
  S : \tau\to\sigma\] such that
\[ \lambda x. T (S x) \betaetaeq \lambda x. x
  \quad\text{and}\quad
  \lambda y. S (T y) \betaetaeq \lambda y. y.\]
Therefore, when $\tau\cong \sigma$, and $\sigma$ happens to be more
canonical than $\tau$---in the sense that to any
$\beta\eta$-equivalence class of type $\tau$ corresponds a smaller one
of type $\sigma$---one can reduce the problem of deciding
$\beta\eta$-equality at $\tau$ to deciding it for a smaller subclass
of terms.
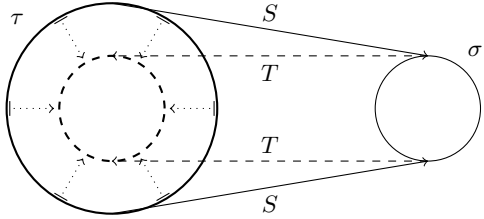
\begin{center}
  \begin{tikzpicture}[scale=0.70]
    \draw [thick] (0,0) circle (2);
    \draw (-1.8,1.7) node {$\tau$};
    \draw [thick,dashed] (0,0) circle (1);
    \foreach \alpha in {0, 60, 120, 180, 240, 300} {
      \draw [<-|,dotted] (\alpha:1.1) -- (\alpha:1.95);
    }
    \draw [->] (0,2) -- (6,1) node[midway,above] {$S$};
    \draw [->] (0,-2) -- (6,-1) node[midway,below] {$S$};
    \draw (6,0) circle (1);
    \draw (6.9,1.1) node {$\sigma$};
    \draw [->,dashed] (6,1) -- (0,1) node[midway,below] {$T$};
    \draw [->,dashed] (6,-1) -- (0,-1) node[midway,above] {$T$};
  \end{tikzpicture}
\end{center}
In the case when $\sigma=\enf{\tau}$, the equivalence classes at type
$\sigma$ will not be larger than their original classes at $\tau$,
since the main effect of the reduction to exp-log normal form is to
get rid of as many sum types on the left of an arrow as possible, and
it is known that for the $\{\times,\to\}$-typed lambda calculus one
can choose a single canonical $\eta$-long $\beta$-normal
representative out of a class of $\beta\eta$-equal terms.

Thus, from the perspective of type isomorphisms, we can observe the
partition of the set of terms of type $\tau$ into
$\betaetaeq$-congruence classes as projected upon different parallel
planes in three dimensional space, one plane for each type isomorphic
to $\tau$. If we choose to observe the planes for $\tau$ and
$\enf{\tau}$, we may describe the situation by the following figure.
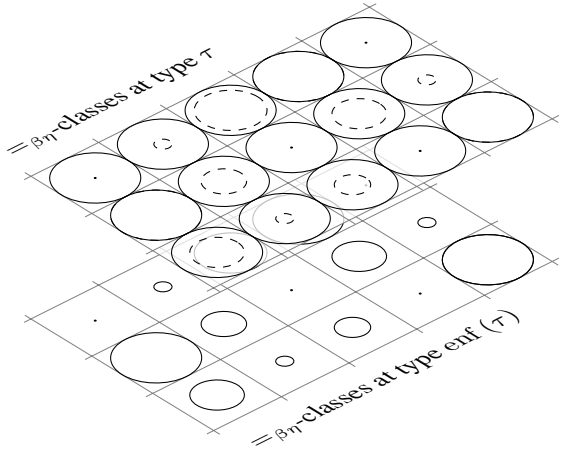
\begin{center}
  \begin{tikzpicture}[scale=0.90]
    \begin{scope}[
      yshift=-60,
      every node/.append style={yslant=\yslant,xslant=\xslant},
      yslant=\yslant,xslant=\xslant
      ]
      \draw[gray, ultra thin] (-0.1,-0.1) grid (5.1,3.1); 
      \draw (0.5,0.5) circle (0.3);
      \draw (1.5,0.5) circle (0.1);
      \draw (2.5,0.5) circle (0.2);
      \draw (3.5,0.5) circle (0.01);
      \draw (4.5,0.5) circle (0.5);
      \draw (4.5,0.5) circle (0.5);
      \draw (0.5,1.5) circle (0.5);
      \draw (1.5,1.5) circle (0.25);
      \draw (2.5,1.5) circle (0.01);
      \draw (3.5,1.5) circle (0.3);
      \draw (4.5,1.5) circle (0.1);
      \draw (0.5,2.5) circle (0.01);
      \draw (1.5,2.5) circle (0.1);
      \draw (2.5,2.5) circle (0.4);
      \draw (3.5,2.5) circle (0.5);
      \draw (4.5,2.5) circle (0.01);
      \fill[black]
      (0.0,-0.5) node[right] {$\betaetaeq$-classes at type $\enf{\tau}$};
    \end{scope}
    \begin{scope}[
      yshift=0,
      every node/.append style={yslant=\yslant,xslant=\xslant},
      yslant=\yslant,xslant=\xslant
      ]
      \fill[white,fill opacity=.75] (0,0) rectangle (5,3); 
      \draw[gray, ultra thin] (-0.1,-0.1) grid (5.1,3.1); 
      \draw (0.5,0.5) circle (0.5);
      \draw[dashed] (0.5,0.5) circle (0.3);
      \draw (1.5,0.5) circle (0.5);
      \draw[dashed] (1.5,0.5) circle (0.1);
      \draw (2.5,0.5) circle (0.5);
      \draw[dashed] (2.5,0.5) circle (0.2);
      \draw (3.5,0.5) circle (0.5);
      \draw[dashed] (3.5,0.5) circle (0.01);
      \draw (4.5,0.5) circle (0.5);
      \draw[dashed] (4.5,0.5) circle (0.5);
      \draw (0.5,1.5) circle (0.5);
      \draw[dashed] (0.5,1.5) circle (0.5);
      \draw (1.5,1.5) circle (0.5);
      \draw[dashed] (1.5,1.5) circle (0.25);
      \draw (2.5,1.5) circle (0.5);
      \draw[dashed] (2.5,1.5) circle (0.01);
      \draw (3.5,1.5) circle (0.5);
      \draw[dashed] (3.5,1.5) circle (0.3);
      \draw (4.5,1.5) circle (0.5);
      \draw[dashed] (4.5,1.5) circle (0.1);
      \draw (0.5,2.5) circle (0.5);
      \draw[dashed] (0.5,2.5) circle (0.01);
      \draw (1.5,2.5) circle (0.5);
      \draw[dashed] (1.5,2.5) circle (0.1);
      \draw (2.5,2.5) circle (0.5);
      \draw[dashed] (2.5,2.5) circle (0.4);
      \draw (3.5,2.5) circle (0.5);
      \draw[dashed] (3.5,2.5) circle (0.5);
      \draw (4.5,2.5) circle (0.5);
      \draw[dashed] (4.5,2.5) circle (0.01);
      \fill[black]
      (0.0,3.5) node[right] {$\betaetaeq$-classes at type $\tau$};
    \end{scope}
  \end{tikzpicture}
\end{center}
The dashed circle depicts the compaction, if any, of a congruence
class achieved by coercing to ENF type. The single point depicts the
compaction to a singleton set, the case where a unique canonical
representative of a class of $\beta\eta$-terms exists.

We do not claim that the plane of $\enf{\tau}$ is \emph{always} the
best possible plane to choose for deciding $\betaetaeq$. Indeed, for
concrete base types there may well be further type isomorphisms to
apply (think of the role of the unit type $1$ in
$(1\to \tau+\sigma)\to\rho$) and hence a better plane than the one for
$\enf{\tau}$. However, it is a reasonably good default choice.

For the cases of types where the sum can be
completely eliminated, such as the two examples of
Section~\ref{sec:introduction}, the projection amounts to compacting
the $\beta\eta$-congruence class to a single point, a canonical normal
term of type $\enf{\tau}$.

Assuming
$\tau,\sigma,\tau_i$ are base types, the canonical representatives
for the two $\beta\eta$-congruence classes of
Section~\ref{sec:introduction} are
\[
\pair{\lambda x. (\fst{(\snd{x})})(\fst{x})}{\lambda x. (\snd{(\snd{x})})(\fst{x})}
\]
and
\[
    \pair{\lambda x. (\fst{x}) ((\fst{\snd{x}})(\fst{\snd{\snd{x}}}))}{\lambda x. (\fst{x}) ((\fst{\snd{x}})(\fst{\snd{\snd{x}}}))}.
\]
Note that, unlike \cite{balat_dicosmo_fiore}, we do not need \emph{any}
sophisticated term analysis to derive a canonical form in this kind of
cases. One may either apply the standard terms witnessing the
isomorphisms by hand, or use our normalizer described in
Section~\ref{sec:terms}.

The natural place to pick a canonical representative is thus the
$\beta\eta$-congruence class of terms at the normal type, not the
class at the original type! Moreover, beware that even if it may be
tempting to map a canonical representative along isomorphic coercions
back to the original type, the obtained representative may not be
truly canonical since there is generally more than one way to specify
the terms $S$ and $T$ that witness a type isomorphism.

Of course, not always can all sum types be eliminated by type
isomorphism, and hence not always can a class be compacted to a single
point in that way. Nevertheless, even in the case where there are
still sums remaining in the type of a term, the ENF simplifies the set
of applicable $=_{\beta\eta}$-axioms.

We can use it to get a restricted set of equations, $\betaetaeqe$, shown in
Figure~\ref{fig:syntax:enf}, which is still complete for proving full
$\beta\eta$-equality, as made precise in the following theorem.
\begin{figure*}
  \centering
  \begin{multline*}
    M,N ::= x^e ~|~ (M^{c\to b} N^c)^b ~|~ (\fst{M^{(c\to b)\times c_0}})^{c\to b} ~|~ (\snd{M^{(c\to b)\times c_0}})^{c_0} ~|~ \ccase{M^{c+d}}{x_1^c}{N_1^e}{x_2^d}{N_2^e}^e \\
    ~|~ (\lambda x^c.M^b)^{c\to b} ~|~ \pair{M^{b\to c}}{N^{c_0}}^{(b\to c)\times c_0} ~|~ (\inl{M^c})^{c+d} ~|~ (\inr{M^d})^{c+d}    
  \end{multline*}
  \begin{align}
    \label{e:beta:arrow}\tag{$\beta_{\to}^\numbere$}(\lambda x^c. N^b) M &\betaeqe N\{M/x\} &\\
    \label{e:beta:product}\tag{$\beta_\times^\numbere$}\pi_i{\pair{M_1^{b\to c}}{M_2^{c_0}}} &\betaeqe M_i & \\
    \label{e:beta:sum}\tag{$\beta_+^\numbere$}\ccase{\iota_i{M}}{x_1}{N_1}{x_2}{N_2}^e &\betaeqe N_i\{M/x_i\} & \\
    \label{e:eta:arrow}\tag{$\eta_{\to}^\numbere$}N^{c\to b} &\etaeqe \lambda x. N x & x\not\in\FV(N) \\
    \label{e:eta:product}\tag{$\eta_\times^\numbere$}N^{(c\to b)\times c_0} &\etaeqe \pair{\fst{N}}{\snd{N}} & \\
    \label{e:eta:sum}\tag{$\eta_+^\numbere$}N^b\{M^d/x\}&\etaeqe \ccase{M}{x_1}{N\{\inl{x_1}/x\}}{x_2}{N\{\inr{x_2}/x\}}^b & x_1,x_2\not\in\FV(N)\\
    \label{e:eta:pi}\tag{$\eta_{\pi}^\numbere$}\pi_i \ccase{M}{x_1}{N_1}{x_2}{N_2} &\etaeqe \ccase{M}{x_1}{\pi_i N_1}{x_2}{\pi_i N_2}^{c} & \\
    \label{e:eta:lambda}\tag{$\eta_{\lambda}^\numbere$}\lambda y. \ccase{M}{x_1}{N_1}{x_2}{N_2} &\etaeqe \ccase{M}{x_1}{\lambda y. N_1}{x_2}{\lambda y. N_2}^{c\to b} & y\not\in\FV(M)
  \end{align}
  \caption{Lambda terms of ENF type and the equational theory
    $\betaetaeqe$.}
  \label{fig:syntax:enf}
\end{figure*}

\begin{theorem} Let $P,Q$ be terms of type $\tau$ and let
  $S:\tau\to\enf{\tau}, T:\enf{\tau}\to\tau$ be a witnessing pair of
  terms for the isomorphism $\tau\cong\enf{\tau}$. Then,
  $P \betaetaeq Q$ if and only if $SP \betaetaeq^{\numbere} SQ$ and if
  and only if $T(SP) \betaetaeq T(SQ)$.
\end{theorem}
\begin{proof}
  Since the set of terms of ENF type is a subset of all typable terms,
  it suffices to show that all $\betaetaeq$-equations that apply to
  terms of ENF type can be derived already by the
  $\betaetaeqe$-equations. 

  Notice first that \ref{e:eta:lambda} and \ref{e:eta:pi} are special
  cases of \ref{eta:sum}, so, in fact, the only axiom missing from
  $\betaetaeqe$ is \ref{eta:sum} itself,
  \begin{multline*}
  N\{M/x\} \etaeqe \ccase{M}{x_1}{N\{\inl{x_1}/x\}}{x_2}{N\{\inr{x_2}/x\}} \\ (x_1,x_2\not\in\FV(N)),
  \end{multline*}
  when $N$ is of type $c$; the case of $N$ of type $d$ is covered
  directly by the \ref{e:eta:sum}-axiom. We thus show that the
  \ref{eta:sum}-axiom is derivable from the $\betaetaeqe$-ones by
  induction on $c$.
  \begin{description}
  \item[Case for $N$ of type $(c\to b)\times c_0$.]
    \begin{align*}
      & N\{M/x\} \\
      \etaeqe & \pair{\fst{(N\{M/x\})}}{\snd{(N\{M/x\})}} \quad \text{by \ref{e:eta:product}} \\
      = & \pair{(\fst{N})\{M/x\}}{(\snd{N})\{M/x\}}  \\
      \etaeqe & \langle{\ccase{M}{x_1}{(\fst{N})\{\inl{x_1}/x\}}{x_2}{(\fst{N})\{\inr{x_2}/x\}}},  \\
              & ~{\ccase{M}{x_1}{(\snd{N})\{\inl{x_1}/x\}}{x_2}{(\snd{N})\{\inr{x_2}/x\}}}\rangle  \quad \text{by IH} \\
      \betaetaeqe & \langle{\fst{(\ccase{M}{x_1}{N\{\inl{x_1}/x\}}{x_2}{N\{\inr{x_2}/x\}})}},  \\
      & ~\snd{(\ccase{M}{x_1}{N\{\inl{x_1}/x\}}{x_2}{N\{\inr{x_2}/x\}})}\rangle  \quad \text{by \ref{e:eta:pi}}\\
      \etaeqe & \ccase{M}{x_1}{N\{\inl{x_1}/x\}}{x_2}{N\{\inr{x_2}/x\}} \quad \text{by \ref{e:eta:product}}
    \end{align*}

  \item[Case for $N$ of type $c\to b$.] 
    \begin{align*}
      & N\{M/x\} \\
      \etaeqe & \lambda y. (N\{M/x\}) y \quad \text{by \ref{e:eta:arrow}}\\
      =~ & \lambda y. (N y)\{M/x\} \quad \text{for } y\not\in\FV{(N\{M/x\})}\\
      \etaeqe & \lambda y. \ccase{M}{x_1}{(Ny)\{\inl{x_1}/x\}}{x_2}{(Ny)\{\inr{x_2}/x\}}  \quad \text{by \ref{e:eta:sum}}\\
      \etaeqe & \ccase{M}{x_1}{(\lambda y. Ny)\{\inl{x_1}/x\}}{x_2}{(\lambda y. Ny)\{\inr{x_2}/x\}}  ~ \text{by \ref{e:eta:lambda}}\\
      \etaeqe & \ccase{M}{x_1}{N\{\inl{x_1}/x\}}{x_2}{N\{\inr{x_2}/x\}}  \quad \text{by \ref{e:eta:arrow}}
    \end{align*}
  \end{description}
\end{proof}

The transformation of terms to ENF type thus allows to simplify the
(up to now) standard axioms of $\betaetaeq$. The new axioms are
complete for $\betaetaeq$ in spite of them being only \emph{special cases}
of the old ones. A notable feature is that we get to disentangle the
left-hand side and right-hand side of the equality axioms: for
instance, the right-hand side of \ref{beta:arrow}-axiom can no longer
overlap with the left-hand side of the \ref{eta:sum}-axiom, due to
typing restrictions on the term $M$.

One could get rid of \ref{e:eta:pi} and \ref{e:eta:lambda} if one had
a version of $\lambda$-calculus resistant to these permuting
conversions. The syntax of such a lambda calculus would further be
simplified if, instead of binary, one had $n$-ary sums and
products. In that case, there would be no need for variables of sum
type at all (currently they can only be introduced by the second
branch of $\delta$). We would in fact get a calculus with only
variables of type $c\to b$, and that would still be suitable as a
small theoretical core of functional programming languages.


\section{A compact representation of terms at ENF type}
\label{sec:terms}

It is the subject of this section to show that the desiderata for a
more canonical calculus from the previous paragraph can in fact be
achieved. We shall define a new representation of lambda terms, that
we have isolated as the most compact syntax possible during the formal
Coq development of a normalizer of terms at ENF type. The description
of the normalizer itself will be left for the second part of this
section, Subsection~\ref{subsec:nbe}. In the first part of the
section, we shall demonstrate the value of representing terms in our
calculus on a number of examples. Comparing our normal form for
\emph{syntactical identity} provides a first such heuristic for
deciding $\betaetaeq$ in the presence of sums.

Before we continue with the presentation of the new calculus, for the
sake of precision, we give the formal representation of terms of the
two term calculi. First, we represent the usual lambda calculus with
sums.
\begin{framed}
\begin{coqdoccode}
\coqdocindent{0.00em}
\coqdockw{Inductive} \coqdocvar{ND} : \coqdocvar{list} \coqdocvar{Formula} \ensuremath{\rightarrow} \coqdocvar{Formula} \ensuremath{\rightarrow} \coqdockw{Set} :=\coqdoceol
\coqdocindent{0.00em}
\ensuremath{|} \coqdocvar{hyp} : \coqdockw{\ensuremath{\forall}} \{\coqdocvar{Gamma} \coqdocvar{A}\},\coqdoceol
\coqdocindent{1.00em}
\coqdocvar{ND} (\coqdocvar{A} :: \coqdocvar{Gamma}) \coqdocvar{A}\coqdoceol
\coqdocindent{0.00em}
\ensuremath{|} \coqdocvar{wkn} : \coqdockw{\ensuremath{\forall}} \{\coqdocvar{Gamma} \coqdocvar{A} \coqdocvar{B}\},\coqdoceol
\coqdocindent{1.00em}
\coqdocvar{ND} \coqdocvar{Gamma} \coqdocvar{A} \ensuremath{\rightarrow} \coqdocvar{ND} (\coqdocvar{B} :: \coqdocvar{Gamma}) \coqdocvar{A}\coqdoceol
\coqdocindent{0.00em}
\ensuremath{|} \coqdocvar{lam} : \coqdockw{\ensuremath{\forall}} \{\coqdocvar{Gamma} \coqdocvar{A} \coqdocvar{B}\},\coqdoceol
\coqdocindent{1.00em}
\coqdocvar{ND} (\coqdocvar{A} :: \coqdocvar{Gamma}) \coqdocvar{B} \ensuremath{\rightarrow} \coqdocvar{ND} \coqdocvar{Gamma} (\coqdocvar{impl} \coqdocvar{A} \coqdocvar{B})\coqdoceol
\coqdocindent{0.00em}
\ensuremath{|} \coqdocvar{app} : \coqdockw{\ensuremath{\forall}} \{\coqdocvar{Gamma} \coqdocvar{A} \coqdocvar{B}\},\coqdoceol
\coqdocindent{1.00em}
\coqdocvar{ND} \coqdocvar{Gamma} (\coqdocvar{impl} \coqdocvar{A} \coqdocvar{B}) \ensuremath{\rightarrow} \coqdocvar{ND} \coqdocvar{Gamma} \coqdocvar{A} \ensuremath{\rightarrow} \coqdocvar{ND} \coqdocvar{Gamma} \coqdocvar{B}\coqdoceol
\coqdocindent{0.00em}
\ensuremath{|} \coqdocvar{pair} : \coqdockw{\ensuremath{\forall}} \{\coqdocvar{Gamma} \coqdocvar{A} \coqdocvar{B}\},\coqdoceol
\coqdocindent{1.00em}
\coqdocvar{ND} \coqdocvar{Gamma} \coqdocvar{A} \ensuremath{\rightarrow} \coqdocvar{ND} \coqdocvar{Gamma} \coqdocvar{B} \ensuremath{\rightarrow} \coqdocvar{ND} \coqdocvar{Gamma} (\coqdocvar{conj} \coqdocvar{A} \coqdocvar{B})\coqdoceol
\coqdocindent{0.00em}
\ensuremath{|} \coqdocvar{fst} : \coqdockw{\ensuremath{\forall}} \{\coqdocvar{Gamma} \coqdocvar{A} \coqdocvar{B}\},\coqdoceol
\coqdocindent{1.00em}
\coqdocvar{ND} \coqdocvar{Gamma} (\coqdocvar{conj} \coqdocvar{A} \coqdocvar{B}) \ensuremath{\rightarrow} \coqdocvar{ND} \coqdocvar{Gamma} \coqdocvar{A}\coqdoceol
\coqdocindent{0.00em}
\ensuremath{|} \coqdocvar{snd} : \coqdockw{\ensuremath{\forall}} \{\coqdocvar{Gamma} \coqdocvar{A} \coqdocvar{B}\},\coqdoceol
\coqdocindent{1.00em}
\coqdocvar{ND} \coqdocvar{Gamma} (\coqdocvar{conj} \coqdocvar{A} \coqdocvar{B}) \ensuremath{\rightarrow} \coqdocvar{ND} \coqdocvar{Gamma} \coqdocvar{B}\coqdoceol
\coqdocindent{0.00em}
\ensuremath{|} \coqdocvar{inl} : \coqdockw{\ensuremath{\forall}} \{\coqdocvar{Gamma} \coqdocvar{A} \coqdocvar{B}\},\coqdoceol
\coqdocindent{1.00em}
\coqdocvar{ND} \coqdocvar{Gamma} \coqdocvar{A} \ensuremath{\rightarrow} \coqdocvar{ND} \coqdocvar{Gamma} (\coqdocvar{disj} \coqdocvar{A} \coqdocvar{B})\coqdoceol
\coqdocindent{0.00em}
\ensuremath{|} \coqdocvar{inr} : \coqdockw{\ensuremath{\forall}} \{\coqdocvar{Gamma} \coqdocvar{A} \coqdocvar{B}\},\coqdoceol
\coqdocindent{1.00em}
\coqdocvar{ND} \coqdocvar{Gamma} \coqdocvar{B} \ensuremath{\rightarrow} \coqdocvar{ND} \coqdocvar{Gamma} (\coqdocvar{disj} \coqdocvar{A} \coqdocvar{B})\coqdoceol
\coqdocindent{0.00em}
\ensuremath{|} \coqdocvar{cas} : \coqdockw{\ensuremath{\forall}} \{\coqdocvar{Gamma} \coqdocvar{A} \coqdocvar{B} \coqdocvar{C}\},\coqdoceol
\coqdocindent{1.00em}
\coqdocvar{ND} \coqdocvar{Gamma} (\coqdocvar{disj} \coqdocvar{A} \coqdocvar{B}) \ensuremath{\rightarrow}\coqdoceol
\coqdocindent{1.00em}
\coqdocvar{ND} (\coqdocvar{A} :: \coqdocvar{Gamma}) \coqdocvar{C} \ensuremath{\rightarrow} \coqdocvar{ND} (\coqdocvar{B} :: \coqdocvar{Gamma}) \coqdocvar{C} \ensuremath{\rightarrow}\coqdoceol
\coqdocindent{1.00em}
\coqdocvar{ND} \coqdocvar{Gamma} \coqdocvar{C}.\coqdoceol
\end{coqdoccode}
  
\end{framed}
\noindent
The constructors are self-explanatory, except for \textit{hyp} and
\textit{wkn}, which are in fact used to denote de Bruijn indices:
\textit{hyp} denotes $0$, while \textit{wkn} is the successor. For
instance, the term $\lambda x y z. y$ is represented as \textit{lam
  (lam (lam (wkn hyp)))} i.e. \textit{lam (lam (lam 1))}.

DeBruijn indices creep in as the simplest way to work with binders in
Coq, and although they may reduce readability, they solve the problem
with $\alpha$-conversion of terms.

Next is our compact representation of terms, defined by the following
simultaneous inductive definition of terms at base type
(\textit{HSb}), together with terms at product type
(\textit{HSc}). These later are simply finite lists of
\textit{HSb}-terms.
\begin{framed}
\begin{coqdoccode}
\coqdocindent{0.00em}
\coqdockw{Inductive} \coqdocvar{HSc} : \coqdocvar{CNF} \ensuremath{\rightarrow} \coqdockw{Set} :=\coqdoceol
\coqdocindent{0.00em}
\ensuremath{|} \coqdocvar{tt} : \coqdocvar{HSc} \coqdocvar{top}\coqdoceol
\coqdocindent{0.00em}
\ensuremath{|} \coqdocvar{pair} : \coqdockw{\ensuremath{\forall}} \{\coqdocvar{c1} \coqdocvar{b} \coqdocvar{c2}\},
\coqdocvar{HSb} \coqdocvar{c1} \coqdocvar{b} \ensuremath{\rightarrow}
\coqdocvar{HSc} \coqdocvar{c2} \ensuremath{\rightarrow} \coqdocvar{HSc} (\coqdocvar{con} \coqdocvar{c1} \coqdocvar{b} \coqdocvar{c2})\coqdoceol
\coqdocindent{0.00em}
\coqdockw{with} \coqdocvar{HSb} : \coqdocvar{CNF} \ensuremath{\rightarrow} \coqdocvar{Base} \ensuremath{\rightarrow} \coqdockw{Set} :=\coqdoceol
\coqdocindent{0.00em}
\ensuremath{|} \coqdocvar{app} : \coqdockw{\ensuremath{\forall}} \{\coqdocvar{p} \coqdocvar{c0} \coqdocvar{c1} \coqdocvar{c2}\},\coqdoceol
\coqdocindent{1.00em}
\coqdocvar{HSc} (\coqdocvar{explogn} \coqdocvar{c1} (\coqdocvar{cnf} (\coqdocvar{ntimes} \coqdocvar{c2} (\coqdocvar{con} \coqdocvar{c1} (\coqdocvar{prp} \coqdocvar{p}) \coqdocvar{c0})))) \ensuremath{\rightarrow}\coqdoceol
\coqdocindent{1.00em}
\coqdocvar{HSb} (\coqdocvar{ntimes} \coqdocvar{c2} (\coqdocvar{con} \coqdocvar{c1} (\coqdocvar{prp} \coqdocvar{p}) \coqdocvar{c0})) (\coqdocvar{prp} \coqdocvar{p})\coqdoceol
\coqdocindent{0.00em}
\ensuremath{|} \coqdocvar{cas} : \coqdockw{\ensuremath{\forall}} \{\coqdocvar{d} \coqdocvar{b} \coqdocvar{c0} \coqdocvar{c1} \coqdocvar{c2} \coqdocvar{c3}\},\coqdoceol
\coqdocindent{1.00em}
\coqdocvar{HSc} (\coqdocvar{explogn} \coqdocvar{c1} (\coqdocvar{cnf} (\coqdocvar{ntimes} \coqdocvar{c2} (\coqdocvar{con} \coqdocvar{c1} (\coqdocvar{bd} \coqdocvar{d}) \coqdocvar{c0})))) \ensuremath{\rightarrow}\coqdoceol
\coqdocindent{1.00em}
\coqdocvar{HSc} (\coqdocvar{explogn} (\coqdocvar{explog0} \coqdocvar{b} \coqdocvar{d})\coqdoceol
\coqdocindent{4.00em}
(\coqdocvar{cnf} (\coqdocvar{ntimes} \coqdocvar{c3} (\coqdocvar{ntimes} \coqdocvar{c2} (\coqdocvar{con} \coqdocvar{c1} (\coqdocvar{bd} \coqdocvar{d}) \coqdocvar{c0}))))) \ensuremath{\rightarrow}\coqdoceol
\coqdocindent{1.00em}
\coqdocvar{HSb} (\coqdocvar{ntimes} \coqdocvar{c3} (\coqdocvar{ntimes} \coqdocvar{c2} (\coqdocvar{con} \coqdocvar{c1} (\coqdocvar{bd} \coqdocvar{d}) \coqdocvar{c0}))) \coqdocvar{b}\coqdoceol
\coqdocindent{0.00em}
\ensuremath{|} \coqdocvar{wkn} : \coqdockw{\ensuremath{\forall}} \{\coqdocvar{c0} \coqdocvar{c1} \coqdocvar{b1} \coqdocvar{b}\},\coqdoceol
\coqdocindent{1.00em}
\coqdocvar{HSb} \coqdocvar{c0} \coqdocvar{b} \ensuremath{\rightarrow} \coqdocvar{HSb} (\coqdocvar{con} \coqdocvar{c1} \coqdocvar{b1} \coqdocvar{c0}) \coqdocvar{b}\coqdoceol
\coqdocindent{0.00em}
\ensuremath{|} \coqdocvar{inl\_two} : \coqdockw{\ensuremath{\forall}} \{\coqdocvar{c0} \coqdocvar{c1} \coqdocvar{c2}\},\coqdoceol
\coqdocindent{1.00em}
\coqdocvar{HSc} (\coqdocvar{explogn} \coqdocvar{c1} (\coqdocvar{cnf} \coqdocvar{c0})) \ensuremath{\rightarrow} \coqdocvar{HSb} \coqdocvar{c0} (\coqdocvar{bd} (\coqdocvar{two} \coqdocvar{c1} \coqdocvar{c2}))\coqdoceol
\coqdocindent{0.00em}
\ensuremath{|} \coqdocvar{inr\_two} : \coqdockw{\ensuremath{\forall}} \{\coqdocvar{c0} \coqdocvar{c1} \coqdocvar{c2}\},\coqdoceol
\coqdocindent{1.00em}
\coqdocvar{HSc} (\coqdocvar{explogn} \coqdocvar{c2} (\coqdocvar{cnf} \coqdocvar{c0})) \ensuremath{\rightarrow} \coqdocvar{HSb} \coqdocvar{c0} (\coqdocvar{bd} (\coqdocvar{two} \coqdocvar{c1} \coqdocvar{c2}))\coqdoceol
\coqdocindent{0.00em}
\ensuremath{|} \coqdocvar{inl\_dis} : \coqdockw{\ensuremath{\forall}} \{\coqdocvar{c0} \coqdocvar{c} \coqdocvar{d}\},\coqdoceol
\coqdocindent{1.00em}
\coqdocvar{HSc} (\coqdocvar{explogn} \coqdocvar{c} (\coqdocvar{cnf} \coqdocvar{c0})) \ensuremath{\rightarrow} \coqdocvar{HSb} \coqdocvar{c0} (\coqdocvar{bd} (\coqdocvar{dis} \coqdocvar{c} \coqdocvar{d}))\coqdoceol
\coqdocindent{0.00em}
\ensuremath{|} \coqdocvar{inr\_dis} : \coqdockw{\ensuremath{\forall}} \{\coqdocvar{c0} \coqdocvar{c} \coqdocvar{d}\},\coqdoceol
\coqdocindent{1.00em}
\coqdocvar{HSb} \coqdocvar{c0} (\coqdocvar{bd} \coqdocvar{d}) \ensuremath{\rightarrow} \coqdocvar{HSb} \coqdocvar{c0} (\coqdocvar{bd} (\coqdocvar{dis} \coqdocvar{c} \coqdocvar{d})).\coqdoceol
\end{coqdoccode}
  
\end{framed}
\noindent

For a more human-readable notation of our calculus, we are going to use
the following one,
\begin{align*}
  P,Q &::= \hspair{M_1}{M_n} \qquad (n\ge 0)\\
  M,M_i &::= \hsappn{n}{P} ~|~ \hscasn{n}{P}{Q} ~|~ \hswkn{M} ~|~ \hsinltwo{P} ~|~ \hsinrtwo{P} ~|~ \hsinldis{P} ~|~ \hsinrdis{M},
\end{align*}
with typing rules as follows:
\begin{gather*}
\frac{M_1:(c_1\vdash b_1) \quad\cdots\quad M_n:(c_n\vdash b_n)}{\hspair{M_1}{M_n} : (c_1\to b_1)\times\cdots\times (c_n\to b_n)} \\~\\
\frac{P:(\explogn{c_1}{\ntimes{c_2}{(c_1\to p) \times c_0}})}{\hsappn{n}{P} : ({\ntimes{c_2}{(c_1\to p)\times c_0}}\vdash{p})}
\end{gather*}
\begin{gather*}
\frac{
  \binomnb{\displaystyle P : (\explogn{c_1}{\ntimes{c_2}{(c_1\to d)\times c_0}})}
  {\displaystyle Q : (\explogn{(\explogone{b}{d})}{\ntimes{c_3}{\ntimes{c_2}{(c_1\to d)\times c_0}}})}
}{\hscasn{n}{P}{Q} : (\ntimes{c_3}{\ntimes{c_2}{(c_1\to d)\times c_0}}\vdash b)}
\end{gather*}
\begin{gather*}
\frac{M : (c_0\vdash b)}{\hswkn{M} : ((c_1\to b_1)\times c_0\vdash b)}
\end{gather*}
\begin{gather*}
\frac{P : (\explogn{c_1}{c_0})}{\inl{P} : (c_0\vdash{c_1+c_2})}\qquad
\frac{P : (\explogn{c_2}{c_0})}{\inl{P} : (c_0\vdash{c_1+c_2})}
\end{gather*}
\begin{gather*}
\frac{P : (\explogn{c}{c_0})}{\iota'_1{P} : (c_0\vdash{c+d})}\qquad
\frac{M : (c_0\vdash d)}{\iota'_2{M} : (c_0\vdash{c+d})}.
\end{gather*}
The typing rules above involve two kinds of typing judgments.
\begin{description}
\item[Judgments at base type:] Denoted $M : (c\vdash b)$, this is the
  main judgment kind, the conclusion of all but the first typing
  rule. It says that $M$ is a term of type $b$ (i.e. either an atomic
  $p$ or a disjunction type $d$) in the typing context $c$. This
  context $c$ takes over the place of the usual context $\Gamma$ and
  allows only hypotheses (variables) of type $c_i\to b_i$ to be used
  inside the term $M$.
\item[Judgments at product type:] Denoted $P : c$ or $Q : c$, this
  kind of judgment is only the conclusion of the first typing rule,
  whose sole purpose is to make a tupple of base type judgments.

  However, the judgments at product type are used as \emph{hypotheses}
  in the other typing rules, where their role is to allow $n$ premises
  to the typing rule. For this usage, they are disguised as the
  macro-expansions $\explogn{c_2}{c_1}$ or
  $\explogn{(\explogone{b}{d})}{c}$. Implemented by the Coq fixpoints
  \textit{explogn} and \textit{explog1}, these macro expansions work
  as follows:
  \begin{multline*}
    \explogn{(c_{1}\to b_{1})\times\cdots\times(c_{n}\to b_{n})}{c_0}\equiv\\
    (c_{1}\times c_0\to b_{1})\times\cdots\times(c_{n}\times c_0\to b_{n})
  \end{multline*}
  \begin{gather*}
    \explogone{b}{(c_1+\cdots+c_n)}\equiv
    (c_1\to b)\times\cdots\times(c_n\to b)
  \end{gather*}
  Note that the usage of $\explogone{b}{d}$ in the typing rule for
  $\delta$ allows the number of premises contained in $Q$ to be
  determined by the size of the sum $d$.
\end{description}
The typing rules also rely on a implicit variable convention, where a
variable $x_n$ actually denotes the variable whose deBruijn index is
$n$ (we start counting from $0$).
\begin{description}
\item[Variables as deBruijn indices:] The variable $x_n$ in the rules
  for $\hsappn{n}{P}$ and $\hscasn{n}{P}{Q}$ represent the hypothesis
  $c_1\to p$ and $c_1\to d$. For concrete $c_2, c_3$, the subscript
  $n$ means that the variable represents the $n$-th hypothesis of the
  form $c\to b$, counting from left to right and starting from 0, in
  the context of the term $P$, or the $n+1$-st, in the context of the
  term $Q$.
\end{description}

We shall motivate our syntax in comparison to the syntax of the lambda
calculus from Figure~\ref{fig:syntax:enf}, by considering in order all
term constructors of the later.




\begin{description}
\item[$x^e$] Since $e\in\ENF$, either $e=c\in\CNF$ or
  $e=d\in\DNF$. Variables of type $d$ only appear as binders in the
  second branch of $\delta$, so if we have $n$-ary instead of binary
  $\delta$'s, the only type a variable $x$ could have will be a
  $c$. But, since $c$ is always of the form
  $(c_1\to b_1)\times\cdots\times(c_n\to b_n)$, a variable $x^c$ could
  be written as a tupple of $n$ variables $x_i$ of types $c_i\to b_i$.
  Moreover, as we want our terms to always be $\eta$-expanded, and
  $c_i\to b_i$ is an arrow type, we will not have a separate syntactic
  category of terms for variables $x_i$ in the new calculus, but they
  will rather be encoded/merged with either the category of
  applications $\hsappn{i}{P}$ (when $b$ is an atomic type $p$), or
  the category of case analysis $\hscasn{i}{P}{Q}$ (when $b\in\DNF$),
  the two new constructors explained below.
\item[$M^{c\to b} N^c$] We shall only need this term constructor at
  type $b=p$, since if $b\in\DNF$ the term $MN$ would not be
  $\eta$-long (we want it to be represented by a
  $\delta(MN, \cdots)$). As we realized during our Coq development, we
  shall only need the case $M=x$, as there will be no other syntactic
  element of type $c\to p$ (there will be no projections $\pi_i$ left,
  while the $\delta$ will only be necessary at type $b$). In
  particular, the application $x \langle\rangle$ can be used to
  represent the old category of variables, where $\langle\rangle$ is
  the empty tupple of unit type $1$ (the nullary product).
\item[$(\fst{M^{(c\to b)\times c_0}})^{c\to b}$] If $M$ is
  $\eta$-expanded, as we want all terms to be, this term would only
  create a $\beta$-redex, and so will not be a part of the new syntax,
  as we are building a syntax for $\beta$-normal and $\eta$-long
  terms.
\item[$(\snd{M^{(c\to b)\times c_0}})^{c_0}$] When product types are
  represented as $n$-ary, the same reasoning as for $\fst{}$ applies,
  so $\snd{}$ will not be part of the new syntax.
\item[$\ccase{M^{c+d}}{x_1^c}{N_1^e}{x_2^d}{N_2^e}^e$] This
  constructor is only needed at the type $e=b$, a consequence of the
  fact that the \ref{e:eta:sum}-axiom is specialized to type $b$: the
  axioms \ref{e:eta:pi} and \ref{e:eta:lambda} will not expressible in
  the new syntax, since it will not contain $\pi_i$, as we saw, and it
  will not contain $\lambda$, as we shall see. We will also only need
  the scrutinee $M$ to be of the form $x N$, like it the case of
  application; this additional restriction was not possible to see
  upfront, but only once we used Coq to analyze the terms needed for
  the normalizer.

  The new constructor $\hscasn{n}{P}{Q}$ is thus like the old
  $\delta(x_n P, \cdots)$, except that $Q$ regroups in the form of an
  $n$-ary tupple all the possible branches of the pattern matching
  (sum types will also be $n$-ary, not binary like before).
\item[$(\lambda x^c.M^b)^{c\to b}$] This terms constructor is already
  severely restricted (for instance only one variable $x$ can be
  abstracted), thanks to the restrictions on the left- and right-hand
  sides of the function type. But, as we found out during the Coq
  development, somewhat to our surprise, there is no need for
  $\lambda$-abstraction in our syntax. When reverse-normalizing from
  our calculus to the standard lambda calculus (see the six examples
  below), $\lambda$'s can be reconstructed thanks to the typing
  information.
\item[$\pair{M^{b\to c}}{N^{c_0}}^{(b\to c)\times c_0}$] This
  constructor will be maintained, corresponding to the only typing
  rule with conclusion a judgment of product type, but it will become
  $n$-ary, $\langle M_1,\ldots,M_n\rangle$. In particular, we may have
  the \emph{nullary} tupple $\hstt$ of the null product type
  (i.e. unit type $1$).
\item[$(\inl{M^c})^{c+d}, (\inr{M^d})^{c+d}$] These constructors will
  be maintained, but will be \emph{duplicated}: the new
  $\hsinltwo{},\hsinrtwo{}$ will only be used to construct a binary
  sum $c_1+c_2$ (this is the base case of sum constructors which must
  be at least binary by construction), while the
  $\hsinldis{},\hsinrdis{}$ will be used to construct sums of the form
  $c+d$.
\end{description}

We shall now show a number of examples that our compact term
representation manages to represent canonically. We will also show
cases when $\beta\eta$-equality can \emph{not} be decided
using bringing terms to the compact normal form. For simplicity, all
type variables ($a,b,c,d,e,f,g,p,q,r,s,i,j,k,l$) are assumed to be of
atomic type, none of them denoting members of \textit{Base},
\textit{CNF}, and \textit{DNF}, anymore, and for the rest of this
subsection.

\begin{convention} We shall adopt the convention of writing the type
  $1\to p$ as $p$ ($1$ is the unit type i.e. the nullary product
  type), writing the application to a nullary pair $\hsappn{n}{\hstt}$
  as $x_n$, and writing a \emph{singleton} pair $\langle M \rangle$ as
  just $M$. Hence, for instance, an application of some term $M$ to a
  singleton pair, containing an application of a term $N$ to a nullary
  pair, $M\langle N \hstt \rangle$, will be written as the more
  readable $M N$ corresponding to the usual $\lambda$-calculus
  intuitions.
\end{convention}

\begin{example} This is the first example from the introduction,
  concerning the relative positions of $\lambda$'s, $\delta$'s, and
  applications. The $\beta\eta$-equal terms
  \begin{gather}
    \lambda x. \lambda y. y \ccase{x}{z}{\inl{z}}{z}{\inr{z}}\\
    \label{ex1:output}\lambda x. \lambda y. \ccase{x}{z}{y (\inl{z})}{z}{y (\inr{z})}\\
    \lambda x. \ccase{x}{z}{\lambda y. y (\inl{z})}{z}{\lambda y. y (\inr{z})}\\
    \lambda x. \lambda y. y x
  \end{gather}
  at type
  \[
    (p+q)\to ((p+q)\to r)\to r,
  \]
  are all normalized to the same canonical representation
  \begin{gather}
    \hspairtwo{\hsappn{0}{\hspairone{\hsappn{2}{}}}}{\hsappn{1}{\hspairone{\hsappn{2}{}}}}
  \end{gather}
  at the ENF type
  \begin{multline*}
    ((p\to r)\times(q\to r)\times p\to r)\times\\((p\to r)\times(q\to
    r)\times q\to r)
  \end{multline*}
  which can be reverse-normalized back to (\ref{ex1:output}). However,
  the point is not that~(\ref{ex1:output}) is somehow better than the
  other 3 terms, but that a canonical representation should be sought
  at the ENF type, not the original type! This remark is valid in
  general, and in particular for the other examples below.
\end{example}


\begin{example}\label{ex:2}
  This is the second example from the introduction (Example~6.2.4.2
  from \cite{balat_dicosmo_fiore}). The $\beta\eta$-equal terms
  \begin{gather}
  \lambda x y z u. x (y z)\label{eq:bdcf1}\\
  \label{ex2:output}\lambda x y z u. \ccase{u}{x_1}{x (y z)}{x_2}{x (y z)}
  \end{gather}
  \begin{multline}
  \lambda x y z u. \ccaseone{u}{x_1}{\ccase{\inl{z}}{y_1}{x (y y_1)}{y_2}{x y_2}}{x_2}{\ccase{\inr{y z}}{y_1}{x (y y_1)}{y_2}{x y_2}}    
  \end{multline}
  \begin{gather}
  \lambda x y z u. \ccase{\ccase{u}{x_1}{\inl{z}}{x_2}{\inr{(y z)}}}{y_1}{x(y y_1)}{y_2}{x y_2}\label{eq:bdcf2}    
  \end{gather}
  at type
  \[
    (a\to b) \to (c\to a) \to c \to (d+e) \to b,
  \]
  are all normalized to the compact term
  \begin{gather}
    \hspairtwo{\hsappn{3}{(\hspairone{\hsappn{2}{\hspairone{\hsappn{1}{}}}})}}{\hsappn{3}{(\hspairone{\hsappn{2}{\hspairone{\hsappn{1}{}}}})}}
  \end{gather}
  at the ENF type
  \begin{multline*}
    (d \times c \times (c\to a) \times (a\to b) \to b)\times\\
    (e \times c \times (c\to a) \times (a\to b) \to b),
  \end{multline*}
  which can then be reverse-normalized to (\ref{ex2:output}), if desired.
\end{example}

A reviewer once remarked that the two previous examples can be handled
just by a CPS transformation. While our implementation \emph{will} be
based on continuations, the reason why these examples are handled by
our method are not continuations, but rather the fact that all sum
types can be eliminated, allowing us to choose a canonical term in the
compact representation of the $\{\to,\times\}$-typed lambda calculus.

\begin{example}[Commuting conversions] The left and right hand sides
  of the common commuting conversions,
  \begin{multline}\label{eq:casapp}
    \lambda x y z u. \ccase{u}{v_1}{y v_1}{v_2}{z v_2} x =_{\beta\eta}\\
    =_{\beta\eta}\lambda x y z u. \ccase{u}{v_1}{(y v_1) x}{v_2}{(z v_2) x},
  \end{multline}
  \begin{multline}\label{eq:cascas}
    \lambda x y z u v. \ccase{\ccase{x}{x_1}{y x_1}{x_2}{z x_2}}{w_1}{u w_1}{w_2}{v w_2} =_{\beta\eta}\\
    =_{\beta\eta}
    \lambda x y z u v. \ccaseone{x}{x_1}{\ccase{y x_1}{w_1}{u w_1}{w_2}{v w_2}}{x_2}{\ccase{z x_2}{w_1}{u w_1}{w_2}{v w_2}}
  \end{multline}
  of types
  \begin{gather*}
    s\to (p\to s\to r)\to (q\to s\to r)\to (p+q)\to r
  \end{gather*}
  and 
  \begin{multline*}
    (p+q) \to (p\to r+s)\to (q\to r+s)\to\\(r \to a)\to (s\to a)\to a,
  \end{multline*}
  are normalized to the compact terms
  \begin{equation}
    \tag{\ref{eq:casapp}'}
    \hspairtwo{\hsappn{2}{\hspairtwo{\hsappn{3}{}}{\hsappn{0}{}}}}{\hsappn{1}{\hspairtwo{\hsappn{3}{}}{\hsappn{0}{}}}},
  \end{equation}
  of ENF type
  \begin{multline*}
    (p\times (s\times q\to r)\times (s\times p\to r)\times s\to r)\times\\
    (q\times (s\times q\to r)\times (s\times p\to r)\times s\to r)
  \end{multline*}
  and
  \begin{gather}
    \tag{\ref{eq:cascas}'}
    \hspairtwo{\hscasn{3}{\hspairone{\hsappn{4}{}}}{\hspairtwo{\hsappn{2}{\hspairone{\hsappn{0}{}}}}{\hsappn{1}{\hspairone{\hsappn{0}{}}}}}}{\hscasn{2}{\hspairone{\hsappn{4}{}}}{\hspairtwo{\hsappn{2}{\hspairone{\hsappn{0}{}}}}{\hsappn{1}{\hspairone{\hsappn{0}{}}}}}},
  \end{gather}
  of ENF type
  \begin{multline*}
        ((s\to a)\times (r \to a)\times (q\to r+s) \times (p\to r+s)\times p \to a)\\
        \times((s\to a)\times (r \to a)\times (q\to r+s) \times (p\to r+s)\times q \to a),
  \end{multline*}
  which can be reverse-normalized to the right-hand sides of
  (\ref{eq:casapp}), and (\ref{eq:cascas}), respectively, if desired.
\end{example}

\begin{example}[Eta equations] Both the left- and the right-hand
  sides of the eta rules (represented as closed terms),
  \begin{gather}
    \label{e:1}\lambda x. x \betaetaeq \lambda x y. x y\\
    \label{e:2}\lambda x. x \betaetaeq \lambda x. \pair{\fst{x}}{\snd{x}}\\
    \label{e:3}\lambda x y. x y \betaetaeq \lambda x y. \ccase{y}{x_1}{x (\inl{x_1})}{x_2}{x (\inr{x_2})}\\
    \label{e:4}\lambda x y z. \ccase{z}{z_1}{\lambda u. x z_1}{z_2}{\lambda u. y z_2} \betaetaeq
    \lambda x y z u. \ccase{z}{z_1}{x z_1}{z_2}{y z_2}\\
    \label{e:5}\lambda x y z. \fst{\ccase{z}{z_1}{{x z_1}}{z_2}{{y z_2}}} \betaetaeq
    \lambda x y z. {\ccase{z}{z_1}{\fst{x z_1}}{z_2}{\fst{y z_2}}}\\
    \label{e:6}\lambda x y z. \snd{\ccase{z}{z_1}{{x z_1}}{z_2}{{y z_2}}} \betaetaeq
    \lambda x y z. {\ccase{z}{z_1}{\snd{x z_1}}{z_2}{\snd{y z_2}}}
  \end{gather}
  of types
  \begin{gather}
    (p\to p)\to (p\to p)\tag{\ref{e:1}}\\
    (p\times q)\to (p\times q)\tag{\ref{e:2}}\\
    ((p+q) \to r)\to ((p+q) \to r)\tag{\ref{e:3}}\\
    (p\to s)\to (q\to s)\to (p+q)\to r\to s\tag{\ref{e:4}}\\
    (p\to s\times r)\to (q\to s\times r)\to (p+q)\to s\tag{\ref{e:5}}\\
    (p\to s\times r)\to (q\to s\times r)\to (p+q)\to r\tag{\ref{e:6}}
  \end{gather}
  are mapped to the same compact term
  \begin{gather}
    \tag{\ref{e:1}'}\hspairone{\hsappn{1}{\hspairone{\hsappn{0}{}}}}\\
    \tag{\ref{e:2}'}\hspairtwo{\hsappn{0}{}}{\hspairone{\hsappn{1}{}}}\\
    \tag{\ref{e:3}'}\hspairtwo{\hsappn{1}{\hspairone{\hsappn{0}{}}}}{\hsappn{2}{\hspairone{\hsappn{0}{}}}}\\
    \tag{\ref{e:4}'}\hspairtwo{\hsappn{3}{\hspairone{\hsappn{1}{}}}}{\hsappn{2}{\hspairone{\hsappn{1}{}}}}\\
    \tag{\ref{e:5}'}\hspairtwo{\hsappn{3}{\hspairone{\hsappn{0}{}}}}{\hsappn{1}{\hspairone{\hsappn{0}{}}}}\\
    \tag{\ref{e:6}'}\hspairtwo{\hsappn{4}{\hspairone{\hsappn{0}{}}}}{\hsappn{2}{\hspairone{\hsappn{0}{}}}},
  \end{gather}
  of ENF types
  \begin{gather*}
    p\times(p\to p)\to p\tag{\ref{e:1}'}\\
    (p\times q\to p)\times(p\times q\to q)\tag{\ref{e:2}'}\\
    (p\times(p\to r)\times(q\to r)\to r)\times\qquad\qquad\qquad\\
    \qquad\qquad\qquad(q\times(p\to r)\times(q\to r)\to r)\tag{\ref{e:3}'}\\
    (r\times p\times (q\to s)\times (p\to s)\to s)\times\qquad\qquad\qquad\\
    \qquad\qquad\qquad(r\times q\times (q\to s)\times (p\to s)\to s)\tag{\ref{e:4}'}\\
    (p\times(q\to s)\times(q\to r)\times(p\to s)\times(p\to r)\to s)\times\qquad\\
    \quad(q\times(q\to s)\times(q\to r)\times(p\to s)\times(p\to r)\to s)\tag{\ref{e:5}'}\\
    (p\times(q\to s)\times(q\to r)\times(p\to s)\times(p\to r)\to r)\times\qquad\\
    \quad(q\times(q\to s)\times(q\to r)\times(p\to s)\times(p\to r)\to r)\tag{\ref{e:6}'},
  \end{gather*}
  and reverse-normalizing these compact terms produces always the
  right-hand side of the corresponding equation involving lambda
  terms.
\end{example}

Finally, as we shall see in the following two examples, our conversion
to compact form does not guarantee a canonical representation for
terms that are equal with respect to the strong forms of
$\beta\eta$-equality used to duplicate subterms
(Example~\ref{example:5}) or change the order of case analysis of
subterms (Example~\ref{example:6}). Although such term transformations
might not be desirable in the setting of real programming languages,
for they change the order of evaluation, in a pure effect-free setting
like a proof assistant, such transformation would be handy to have.

\begin{example}\label{example:5} The following $\beta\eta$-equal terms,
  \begin{gather}
    \label{ex5:1}\lambda x y z u. \ccase{u z}{w}{x w}{w}{y w}\\
    \label{ex5:2}\lambda x y z u. \ccase{u z}{w}{\ccase{uz}{w'}{x w'}{w'}{y w'}}{w}{y w},
  \end{gather}
  of type
  \[
    (f\to g)\to (h\to g)\to i\to (i\to f+h)\to g
  \]
  are normalized to two \emph{different} compact representations:
  \begin{gather}
    \tag{\ref{ex5:1}'}
    \hspairone{\hscasn{0}{\hspairone{\hsappn{1}{}}}{\hspairtwo{\hsappn{4}{\hspairone{\hsappn{0}{}}}}{\hsappn{3}{\hspairone{\hsappn{0}{}}}}}}\\
    \tag{\ref{ex5:2}'}
    \hspairone{\hscasn{0}{\hspairone{\hsappn{1}{}}}{\hspairtwo{\hscasn{1}{\hspairone{\hsappn{2}{}}}{\hspairtwo{\hsappn{5}{\hspairone{\hsappn{0}{}}}}{\hsappn{4}{\hspairone{\hsappn{0}{}}}}}}{\hsappn{3}{\hspairone{\hsappn{0}{}}}}}},
  \end{gather}
  of ENF type
  \[
    (i\to f+h)\times i \times (h\to g)\times (f\to g)\to g
  \]
  which can then be reverse-normalized to the starting lambda terms
  themselves.
\end{example}

\begin{example}\label{example:6} The following $\beta\eta$-equal terms,
  \begin{gather}
    \label{ex6:1}\lambda x y z u v. \ccase{z v}{x_1}{\inl{x}}{x_2}{\ccase{u v}{y_1}{\inr{y}}{y_2}{\inl{x}}}\\
    \label{ex6:2}\lambda x y z u v. \ccase{u v}{y_1}{\ccase{z v}{x_1}{\inl{x}}{x_2}{\inr{y}}}{y_2}{\inl{x}},
  \end{gather}
  of type
  \[
    k\to l\to (f\to g+h)\to (f\to i+j)\to f\to k+l
  \]
  are normalized to two \emph{different} compact representations:
  \begin{gather}
    \tag{\ref{ex6:1}'}
    \langle\hscasn{2}{\hspairone{\hsappn{0}{}}}
    {\hspairtwo{\hsinltwo{\hspairone{\hsappn{5}{}}}}{\hspairone{\hscasn{3}{\hspairone{\hsappn{1}{}}}{\hspairtwo{\hsinrtwo{\hspairone{\hsappn{5}{}}}}{\hsinltwo{\hspairone{\hsappn{6}{}}}}}}}}
    \rangle\\
    \tag{\ref{ex6:2}'}
    \langle\hscasn{1}{\hspairone{\hsappn{0}{}}}
    {\hspairtwo{\hscasn{2}{\hspairone{\hsappn{1}{}}}{\hspairtwo{\hsinltwo{\hspairone{\hsappn{6}{}}}}{\hsinrtwo{\hspairone{\hsappn{5}{}}}}}}{\hsinltwo{\hspairone{\hsappn{5}{}}}}}
    \rangle,
  \end{gather}
  of ENF type
  \[
    f\times (f\to i+j)\times (f\to g+h)\times l \times k\to k+l
  \]
  which can then be reverse-normalized to the starting lambda terms
  themselves.
\end{example}

\paragraph{Comparison to the examples covered by the heuristic of
  \cite{balat_dicosmo_fiore}} In addition to Example~\ref{ex:2} that
was borrowed from \cite{balat_dicosmo_fiore}, other examples that can
be covered from that paper are examples~4.2.1--~4.2.4 and
Example~4.3.1. In these examples, not only are the input and the
output of their TDPE represented uniquely, but also, in the cases when
there are two distinct output normal forms according to Balat et al.,
our normalizer unifies the two normal forms into one, shown below:
\begin{gather}
  \hspairone{\hsappn{0}{}}\tag{4.2.1}\\
  \hspairtwo{\hsinltwo{\hspairone{\hsappn{0}{}}}}{\hsinrtwo{\hspairone{\hsappn{0}{}}}}\tag{4.2.2}\\
  \hspairfour{\hsinldis{\hspairtwo{\hsappn{0}{}}{\hsappn{1}{}}}}{\hsinrdis{\hsinldis{\hspairtwo{\hsappn{0}{}}{\hsappn{1}{}}}}}{\hsinrdis{\hsinrdis{\hsinltwo{\hspairtwo{\hsappn{0}{}}{\hsappn{1}{}}}}}}{\hsinrdis{\hsinrdis{\hsinrtwo{\hspairtwo{\hsappn{0}{}}{\hsappn{1}{}}}}}}\tag{4.2.3}\\
  \hspairfour{\hsinldis{\hspairtwo{\hsappn{1}{}}{\hsappn{0}{}}}}{\hsinrdis{\hsinrdis{\hsinltwo{\hspairtwo{\hsappn{1}{}}{\hsappn{0}{}}}}}}{\hsinrdis{\hsinldis{{\hspairtwo{\hsappn{1}{}}{\hsappn{0}{}}}}}}{\hsinrdis{\hsinrdis{\hsinrtwo{\hspairtwo{\hsappn{1}{}}{\hsappn{0}{}}}}}}\tag{4.2.4}\\
  \hspairtwo{\hsinrdis{\hsinrdis{\hsinltwo{\hspairone{\hsappn{1}{}}}}}}{\hsinrdis{\hsinldis{{\hspairone{\hsappn{1}{}}}}}}.\tag{4.3.1}
\end{gather}
Canonical representations could be obtained in these examples, because
it was possible to represent the input and output terms in the
fragment of the compact calculus which does not include $\delta$'s
(although it still involves sum types).

On the other hand, there are also the examples where the input and
output are \emph{not} unified by our procedure.\footnote{We shall not
  reproduce the compact representation for these examples in this
  paper, but they are available for inspection in the Coq
  formalization accompanying it.}  Examples~4.3.2 and~4.3.3 are not
handled because we do not permute the order of case analyses (as shown
by our Example~\ref{example:6}); Example~6.2.4.2 is not handled
because we do not analyze if a subterm has been used twice in a term
or not (as shown also by our Example~\ref{example:5}); examples~4.3.4
and~4.4 are not even executable in our implementation, because we do
not have a special treatment of the atomic empty type.

Of course, there is nothing stopping us from applying the program
transformations that would allow to handle this kind of cases---or
nothing stopping Balat et al. from first applying our type-directed
normalization procedure before performing their heuristic to unify the
different normal outputs that they sometimes get. The point is that
these two methodologies are orthogonal and they would ideally be used
in combination inside a real-world application; for more comments
about the two approaches, see Section~\ref{sec:conclusion}.

\subsection{A converter for the compact term representation}
\label{subsec:nbe}

In the remaining part of this section, we explain the high-level
structure of our prototype normalizer of lambda terms into compact
terms and vice versa. The full Coq implementation of the normalizer,
together with the examples considered above, is given as a companion
to this paper. This is only one possible implementation, using
continuations, but all the previous material of this paper was written
as generically as possible, so that it is useful if other
implementation techniques are attempted in the future, such as
rewriting based on evaluation contexts (i.e. the first-order
reification of continuations), or abstract machines.

In a nutshell, our implementation is a type-directed partial
evaluator, written in continuation-passing style, with an intermediate
phase between the evaluation and reification phases, that allows to
map a `semantic' representation of a term from a type to its ENF type,
and vice versa.
Such partial evaluators can be implemented very elegantly, and with
getting certain correctness properties for free, using the GADTs from
Ocaml's type system, as shown by the recent work of Danvy, Keller, and
Puech~\cite{DanvyKP2015}. Nevertheless, we had chosen to carry out our
implementation in Coq, because that allowed us to perform a careful
interactive analysis of the necessary normal forms---hence the compact
calculus introduced in the first part of this section.

Usual type-directed partial evaluation (TDPE), aka
normalization-by-evaluation (NBE), proceeds in tho phases. First an
evaluator is defined which takes the input term and obtains its
semantic representation, and then a reifier is used to map the
semantic representation into an output syntactic term. Our TDPE uses
an \emph{intermediate} phase between the two phases, a phase where
type isomorphisms are applied to the semantic domain so that the
narrowing down of a class of equal terms, described in
Section~\ref{sec:types}, is performed on the semantic annotation of a
term.

The semantics that we use is defined by a continuation monad over a
\emph{forcing structure}, together with \emph{forcing fixpoints} that
map the type of the input term into a type of the ambient type theory.
The forcing structure is an abstract signature (Coq module type),
requiring a set \textit{K} of possible worlds, a preorder relation on
worlds, \textit{le}, an interpretation of atomic types,
\textit{pforces}, and \textit{X}, the return type of the continuation
monad.
\begin{framed}
\begin{coqdoccode}
\coqdocnoindent
\coqdockw{Module} \coqdockw{Type} \coqdocvar{ForcingStructure}.\coqdoceol
\coqdocindent{1.00em}
\coqdockw{Parameter} \coqdocvar{K} : \coqdockw{Set}.\coqdoceol
\coqdocindent{1.00em}
\coqdockw{Parameter} \coqdocvar{le} : \coqdocvar{K} \ensuremath{\rightarrow} \coqdocvar{K} \ensuremath{\rightarrow} \coqdockw{Set}.\coqdoceol
\coqdocindent{1.00em}
\coqdockw{Parameter} \coqdocvar{pforces} : \coqdocvar{K} \ensuremath{\rightarrow} \coqdocvar{Proposition} \ensuremath{\rightarrow} \coqdockw{Set}.\coqdoceol
\coqdocindent{1.00em}
\coqdockw{Parameter} \coqdocvar{Answer} : \coqdockw{Set}.\coqdoceol
\coqdocindent{1.00em}
\coqdockw{Parameter} \coqdocvar{X}  : \coqdocvar{K} \ensuremath{\rightarrow} \coqdocvar{Answer} \ensuremath{\rightarrow} \coqdockw{Set}.\coqdoceol
\coqdocnoindent
\coqdockw{End} \coqdocvar{ForcingStructure}.\coqdoceol
\end{coqdoccode}


\end{framed}
The continuation monad is polymorphic and instantiable by a forcing
fixpoint \textit{f} and a world \textit{w}. It ensures that the
preorder relation is respected; intuitively, this has to do with
preserving the monotonicity of context free variables: we cannot
`forget' a free variable i.e. contexts cannot decrease.
\begin{framed}
\begin{coqdoccode}
\coqdocindent{0.00em}
\coqdockw{Definition} \coqdocvar{Cont} \{\coqdocvar{class}:\coqdockw{Set}\}(\coqdocvar{f}:\coqdocvar{K}\ensuremath{\rightarrow}\coqdocvar{class}\ensuremath{\rightarrow}\coqdockw{Set})(\coqdocvar{w}:\coqdocvar{K})(\coqdocvar{x}:\coqdocvar{class})\coqdoceol
\coqdocindent{1.00em}
 := \coqdockw{\ensuremath{\forall}} (\coqdocvar{x0}:\coqdocvar{Answer}), \coqdockw{\ensuremath{\forall}} \{\coqdocvar{w'}\},
\coqdocvar{le} \coqdocvar{w} \coqdocvar{w'} \ensuremath{\rightarrow}\coqdoceol
\coqdocindent{3.00em}
(\coqdockw{\ensuremath{\forall}} \{\coqdocvar{w'{}'}\}, \coqdocvar{le} \coqdocvar{w'} \coqdocvar{w'{}'} \ensuremath{\rightarrow} \coqdocvar{f} \coqdocvar{w'{}'} \coqdocvar{x} \ensuremath{\rightarrow} \coqdocvar{X} \coqdocvar{w'{}'} \coqdocvar{x0}) \ensuremath{\rightarrow}
\coqdocvar{X} \coqdocvar{w'} \coqdocvar{x0}.\coqdoceol
\end{coqdoccode}

\end{framed}

Next, the necessary forcing fixpoints are defined: \textit{bforces},
\textit{cforces}, and \textit{dforces}, which are used to construct the type
of the continuation monad corresponding to \textit{Base},
\textit{CNF}, and \textit{DNF}, respectively; \textit{sforces} is used
for constructing the type of the continuation monad corresponding to
non-normalized types.
\begin{framed}
\begin{coqdoccode}
\coqdocindent{0.00em}
\coqdockw{Fixpoint} \coqdocvar{bforces} (\coqdocvar{w}:\coqdocvar{K})(\coqdocvar{b}:\coqdocvar{Base}) \{\coqdockw{struct} \coqdocvar{b}\} : \coqdockw{Set} :=\coqdoceol
\coqdocindent{1.00em}
\coqdockw{match} \coqdocvar{b} \coqdockw{with}\coqdoceol
\coqdocindent{1.00em}
\ensuremath{|} \coqdocvar{prp} \coqdocvar{p} \ensuremath{\Rightarrow} \coqdocvar{pforces} \coqdocvar{w} \coqdocvar{p}\coqdoceol
\coqdocindent{1.00em}
\ensuremath{|} \coqdocvar{bd} \coqdocvar{d} \ensuremath{\Rightarrow} \coqdocvar{dforces} \coqdocvar{w} \coqdocvar{d}\coqdoceol
\coqdocindent{1.00em}
\coqdockw{end}\coqdoceol
\coqdocindent{0.00em}
\coqdockw{with} 
\coqdocvar{cforces} (\coqdocvar{w}:\coqdocvar{K})(\coqdocvar{c}:\coqdocvar{CNF}) \{\coqdockw{struct} \coqdocvar{c}\} : \coqdockw{Set} :=\coqdoceol
\coqdocindent{1.00em}
\coqdockw{match} \coqdocvar{c} \coqdockw{with}\coqdoceol
\coqdocindent{1.00em}
\ensuremath{|} \coqdocvar{top} \ensuremath{\Rightarrow} \coqdocvar{unit}\coqdoceol
\coqdocindent{1.00em}
\ensuremath{|} \coqdocvar{con} \coqdocvar{c1} \coqdocvar{b} \coqdocvar{c2} \ensuremath{\Rightarrow}\coqdoceol
\coqdocindent{1.50em}
(\coqdockw{\ensuremath{\forall}} \coqdocvar{w'}, \coqdocvar{le} \coqdocvar{w} \coqdocvar{w'} \ensuremath{\rightarrow} \coqdocvar{Cont} \coqdocvar{cforces} \coqdocvar{w'} \coqdocvar{c1} \ensuremath{\rightarrow} \coqdocvar{Cont} \coqdocvar{bforces} \coqdocvar{w'} \coqdocvar{b})\coqdoceol
\coqdocindent{1.50em}
\ensuremath{\times} (\coqdocvar{Cont} \coqdocvar{cforces} \coqdocvar{w} \coqdocvar{c2})\coqdoceol
\coqdocindent{1.00em}
\coqdockw{end}\coqdoceol
\coqdocindent{0.00em}
\coqdockw{with} 
\coqdocvar{dforces} (\coqdocvar{w}:\coqdocvar{K})(\coqdocvar{d}:\coqdocvar{DNF}) \{\coqdockw{struct} \coqdocvar{d}\} : \coqdockw{Set} :=\coqdoceol
\coqdocindent{1.00em}
\coqdockw{match} \coqdocvar{d} \coqdockw{with}\coqdoceol
\coqdocindent{1.00em}
\ensuremath{|} \coqdocvar{two} \coqdocvar{c1} \coqdocvar{c2} \ensuremath{\Rightarrow} (\coqdocvar{Cont} \coqdocvar{cforces} \coqdocvar{w} \coqdocvar{c1}) + (\coqdocvar{Cont} \coqdocvar{cforces} \coqdocvar{w} \coqdocvar{c2})\coqdoceol
\coqdocindent{1.00em}
\ensuremath{|} \coqdocvar{dis} \coqdocvar{c1} \coqdocvar{d2} \ensuremath{\Rightarrow} (\coqdocvar{Cont} \coqdocvar{cforces} \coqdocvar{w} \coqdocvar{c1}) + (\coqdocvar{Cont} \coqdocvar{dforces} \coqdocvar{w} \coqdocvar{d2})\coqdoceol
\coqdocindent{1.00em}
\coqdockw{end}.\coqdoceol
\coqdocemptyline
\coqdocindent{0.00em}
\coqdockw{Fixpoint} \coqdocvar{eforces} (\coqdocvar{w}:\coqdocvar{K})(\coqdocvar{e}:\coqdocvar{ENF}) \{\coqdockw{struct} \coqdocvar{e}\} : \coqdockw{Set} :=\coqdoceol
\coqdocindent{1.00em}
\coqdockw{match} \coqdocvar{e} \coqdockw{with}\coqdoceol
\coqdocindent{1.00em}
\ensuremath{|} \coqdocvar{cnf} \coqdocvar{c} \ensuremath{\Rightarrow} \coqdocvar{cforces} \coqdocvar{w} \coqdocvar{c}\coqdoceol
\coqdocindent{1.00em}
\ensuremath{|} \coqdocvar{dnf} \coqdocvar{d} \ensuremath{\Rightarrow} \coqdocvar{dforces} \coqdocvar{w} \coqdocvar{d}\coqdoceol
\coqdocindent{1.00em}
\coqdockw{end}.\coqdoceol
\coqdocemptyline
\coqdocindent{0.00em}
\coqdockw{Fixpoint} \coqdocvar{sforces} (\coqdocvar{w}:\coqdocvar{K})(\coqdocvar{F}:\coqdocvar{Formula}) \{\coqdockw{struct} \coqdocvar{F}\} : \coqdockw{Set} :=\coqdoceol
\coqdocindent{1.00em}
\coqdockw{match} \coqdocvar{F} \coqdockw{with}\coqdoceol
\coqdocindent{1.00em}
\ensuremath{|} \coqdocvar{prop} \coqdocvar{p} \ensuremath{\Rightarrow} \coqdocvar{pforces} \coqdocvar{w} \coqdocvar{p}\coqdoceol
\coqdocindent{1.00em}
\ensuremath{|} \coqdocvar{disj} \coqdocvar{F} \coqdocvar{G} \ensuremath{\Rightarrow} (\coqdocvar{Cont} \coqdocvar{sforces} \coqdocvar{w} \coqdocvar{F}) + (\coqdocvar{Cont} \coqdocvar{sforces} \coqdocvar{w} \coqdocvar{G})\coqdoceol
\coqdocindent{1.00em}
\ensuremath{|} \coqdocvar{conj} \coqdocvar{F} \coqdocvar{G} \ensuremath{\Rightarrow} (\coqdocvar{Cont} \coqdocvar{sforces} \coqdocvar{w} \coqdocvar{F}) \ensuremath{\times} (\coqdocvar{Cont} \coqdocvar{sforces} \coqdocvar{w} \coqdocvar{G})\coqdoceol
\coqdocindent{1.00em}
\ensuremath{|} \coqdocvar{impl} \coqdocvar{F} \coqdocvar{G} \ensuremath{\Rightarrow} \coqdockw{\ensuremath{\forall}} \coqdocvar{w'},\coqdoceol
\coqdocindent{1.50em}
\coqdocvar{le} \coqdocvar{w} \coqdocvar{w'} \ensuremath{\rightarrow} (\coqdocvar{Cont} \coqdocvar{sforces} \coqdocvar{w'} \coqdocvar{F}) \ensuremath{\rightarrow} (\coqdocvar{Cont} \coqdocvar{sforces} \coqdocvar{w'} \coqdocvar{G})\coqdoceol
\coqdocindent{1.00em}
\coqdockw{end}.\coqdoceol
\end{coqdoccode}


\end{framed}

Given these definitions, we can write an evaluator for compact terms,
actually two simultaneously defined evaluators \textit{evalc} and
\textit{evalb}, proceeding by induction on the input term.
\begin{framed}
\begin{coqdoccode}
\coqdocindent{0.00em}
\coqdockw{Theorem} \coqdocvar{evalc} \{\coqdocvar{c}\} : (\coqdocvar{HSc} \coqdocvar{c} \ensuremath{\rightarrow} \coqdockw{\ensuremath{\forall}} \{\coqdocvar{w}\}, \coqdocvar{Cont} \coqdocvar{cforces} \coqdocvar{w} \coqdocvar{c})\coqdoceol
\coqdocindent{0.00em}
\coqdockw{with} \coqdocvar{evalb} \{\coqdocvar{b} \coqdocvar{c0}\} : (\coqdocvar{HSb} \coqdocvar{c0} \coqdocvar{b} \ensuremath{\rightarrow} \coqdockw{\ensuremath{\forall}} \{\coqdocvar{w}\},\coqdoceol
\coqdocindent{8.00em}
\coqdocvar{Cont} \coqdocvar{cforces} \coqdocvar{w} \coqdocvar{c0} \ensuremath{\rightarrow} \coqdocvar{Cont} \coqdocvar{bforces} \coqdocvar{w} \coqdocvar{b}).\coqdoceol
\end{coqdoccode}


\end{framed}

An evaluator for usual lambda terms can also be defined, by induction
on the input term. A helper function \textit{lforces} analogous to the
list map function for \textit{sforces} is necessary.
\begin{framed}
\begin{coqdoccode}
\coqdocindent{0.00em}
\coqdockw{Fixpoint}\coqdoceol
\coqdocindent{1.00em}
\coqdocvar{lforces} (\coqdocvar{w}:\coqdocvar{K})(\coqdocvar{Gamma}:\coqdocvar{list} \coqdocvar{Formula}) \{\coqdockw{struct} \coqdocvar{Gamma}\} : \coqdockw{Set} :=\coqdoceol
\coqdocindent{1.00em}
\coqdockw{match} \coqdocvar{Gamma} \coqdockw{with}\coqdoceol
\coqdocindent{1.00em}
\ensuremath{|} \coqdocvar{nil} \ensuremath{\Rightarrow} \coqdocvar{unit}\coqdoceol
\coqdocindent{1.00em}
\ensuremath{|} \coqdocvar{cons} \coqdocvar{A} \coqdocvar{Gamma0} \ensuremath{\Rightarrow}\coqdoceol
\coqdocindent{2.00em}
\coqdocvar{Cont} \coqdocvar{sforces} \coqdocvar{w} \coqdocvar{A} \ensuremath{\times} \coqdocvar{lforces} \coqdocvar{w} \coqdocvar{Gamma0}\coqdoceol
\coqdocindent{1.00em}
\coqdockw{end}.\coqdoceol
\coqdocemptyline
\coqdocindent{0.00em}
\coqdockw{Theorem} \coqdocvar{eval} \{\coqdocvar{A} \coqdocvar{Gamma}\} : \coqdocvar{ND} \coqdocvar{Gamma} \coqdocvar{A} \ensuremath{\rightarrow} \coqdockw{\ensuremath{\forall}} \{\coqdocvar{w}\},\coqdoceol
\coqdocindent{3.00em}
\coqdocvar{lforces} \coqdocvar{w} \coqdocvar{Gamma} \ensuremath{\rightarrow} \coqdocvar{Cont} \coqdocvar{sforces} \coqdocvar{w} \coqdocvar{A}.\coqdoceol
\end{coqdoccode}


\end{framed}
The novelty of our implementation (besides isolating the compact term
calculus itself), in comparison to previous type-directed partial
evaluators for the lambda calculus with sums, consists in showing that
one can go back and forth between the semantic annotation at a type
$F$ and the semantic annotation of the normal form $\enf{F}$. The
proof of this statement needs a number of auxiliary lemmas that we do
not mention in the paper. We actually prove two statements
simultaneously, \textit{f2f} and \textit{f2f'}, declared as follows.
\begin{framed}
\begin{coqdoccode}
\coqdocindent{0.00em}
\coqdockw{Theorem} \coqdocvar{f2f} :\coqdoceol
\coqdocindent{1.00em}
(\coqdockw{\ensuremath{\forall}} \coqdocvar{F}, \coqdockw{\ensuremath{\forall}} \coqdocvar{w}, \coqdocvar{Cont} \coqdocvar{sforces} \coqdocvar{w} \coqdocvar{F} \ensuremath{\rightarrow} \coqdocvar{Cont} \coqdocvar{eforces} \coqdocvar{w} (\coqdocvar{enf} \coqdocvar{F}))\coqdoceol
\coqdocindent{0.00em}
\coqdockw{with} \coqdocvar{f2f'} : \coqdoceol
\coqdocindent{1.00em}
(\coqdockw{\ensuremath{\forall}} \coqdocvar{F}, \coqdockw{\ensuremath{\forall}} \coqdocvar{w}, \coqdocvar{Cont} \coqdocvar{eforces} \coqdocvar{w} (\coqdocvar{enf} \coqdocvar{F}) \ensuremath{\rightarrow} \coqdocvar{Cont} \coqdocvar{sforces} \coqdocvar{w} \coqdocvar{F}).\coqdoceol
\end{coqdoccode}


\end{framed}
\noindent
As one can see from their type signatures, \textit{f2f} and
\textit{f2f'} provide a link between the semantics of the standard
lambda calculus for sums (\textit{ND}) and the semantics of our
compact calculus(\textit{HSc}/\textit{HSb}).

We move forward to describing the reification phase. In this phase,
two instantiations of a forcing structure are needed. Unlike the
evaluators, which can work over an abstract forcing structure, the
reifiers need concrete instantiations built from the syntax of the
term calculus in order to produce syntactic normal forms.

The first instantiation is a forcing structure for the standard lambda
calculus with sums. The set of worlds is the set of contexts (lists of
types), the preorder on worlds is defined as the prefix relation on
contexts, the forcing of an atomic type $p$ is the set of terms of
type $p$ in the context $w$, and the answer type of the continuation
monad is the set of terms of type $F$ in the context $w$. One could be
more precise, and instantiate the answer type by the set of
\emph{normal/neutral} terms, like it has been done in most other
implementations of TDPE, and in our own prior works, but for the sake
of simplicity, we do not make that distinction in this paper.
\begin{framed}
\begin{coqdoccode}
\coqdocnoindent
\coqdockw{Module} \coqdocvar{structureND} \textless: \coqdocvar{ForcingStructure}.\coqdoceol
\coqdocindent{1.00em}
\coqdockw{Definition} \coqdocvar{K} := \coqdocvar{list} \coqdocvar{Formula}.\coqdoceol
\coqdocemptyline
\coqdocindent{1.00em}
\coqdockw{Inductive} \coqdocvar{le\_} : \coqdocvar{list} \coqdocvar{Formula} \ensuremath{\rightarrow} \coqdocvar{list} \coqdocvar{Formula} \ensuremath{\rightarrow} \coqdockw{Set} :=\coqdoceol
\coqdocindent{1.00em}
\ensuremath{|} \coqdocvar{le\_\_refl} : \coqdockw{\ensuremath{\forall}} \{\coqdocvar{w}\}, \coqdocvar{le\_} \coqdocvar{w} \coqdocvar{w}\coqdoceol
\coqdocindent{1.00em}
\ensuremath{|} \coqdocvar{le\_\_cons} : \coqdockw{\ensuremath{\forall}} \{\coqdocvar{w1} \coqdocvar{w2} \coqdocvar{F}\},\coqdoceol
\coqdocindent{3.00em}
\coqdocvar{le\_} \coqdocvar{w1} \coqdocvar{w2} \ensuremath{\rightarrow} \coqdocvar{le\_} \coqdocvar{w1} (\coqdocvar{cons} \coqdocvar{F} \coqdocvar{w2}).\coqdoceol
\coqdocemptyline
\coqdocindent{1.00em}
\coqdockw{Definition} \coqdocvar{le} := \coqdocvar{le\_}.\coqdoceol
\coqdocemptyline
\coqdocindent{1.00em}
\coqdockw{Definition} \coqdocvar{le\_refl} : \coqdockw{\ensuremath{\forall}} \{\coqdocvar{w}\}, \coqdocvar{le} \coqdocvar{w} \coqdocvar{w}.\coqdoceol
\coqdocemptyline
\coqdocindent{1.00em}
\coqdockw{Definition} \coqdocvar{pforces} := \coqdockw{fun} \coqdocvar{w} \coqdocvar{p} \ensuremath{\Rightarrow} \coqdocvar{ND} \coqdocvar{w} (\coqdocvar{prop} \coqdocvar{p}).\coqdoceol
\coqdocemptyline
\coqdocindent{1.00em}
\coqdockw{Definition} \coqdocvar{Answer} := \coqdocvar{Formula}.\coqdoceol
\coqdocemptyline
\coqdocindent{1.00em}
\coqdockw{Definition} \coqdocvar{X} :=  \coqdockw{fun} \coqdocvar{w} \coqdocvar{F} \ensuremath{\Rightarrow} \coqdocvar{ND} \coqdocvar{w} \coqdocvar{F}.\coqdoceol
\coqdocnoindent
\coqdockw{End} \coqdocvar{structureND}.\coqdoceol
\end{coqdoccode}


\end{framed}

The second instantiation is a forcing structure for our calculus of
compact terms. The set of worlds is the same as the set of CNFs,
because our context are simply CNFs, the preorder is the prefix
relation on CNFs, the forcing of atomic types are terms of atomic
types, and the answer type of the continuation monad is the set of
terms at base type.
\begin{framed}
\begin{coqdoccode}
\coqdocnoindent
\coqdockw{Module} \coqdocvar{structureHS} \textless: \coqdocvar{ForcingStructure}.\coqdoceol
\coqdocindent{1.00em}
\coqdockw{Definition} \coqdocvar{K} := \coqdocvar{CNF}.\coqdoceol
\coqdocemptyline
\coqdocindent{1.00em}
\coqdockw{Inductive} \coqdocvar{le\_} : \coqdocvar{CNF} \ensuremath{\rightarrow} \coqdocvar{CNF} \ensuremath{\rightarrow} \coqdockw{Set} :=\coqdoceol
\coqdocindent{1.00em}
\ensuremath{|} \coqdocvar{le\_\_refl} : \coqdockw{\ensuremath{\forall}} \{\coqdocvar{w}\}, \coqdocvar{le\_} \coqdocvar{w} \coqdocvar{w}\coqdoceol
\coqdocindent{1.00em}
\ensuremath{|} \coqdocvar{le\_\_cons} : \coqdockw{\ensuremath{\forall}} \{\coqdocvar{w1} \coqdocvar{w2} \coqdocvar{c} \coqdocvar{b}\},\coqdoceol
\coqdocindent{3.00em}
\coqdocvar{le\_} \coqdocvar{w1} \coqdocvar{w2} \ensuremath{\rightarrow} \coqdocvar{le\_} \coqdocvar{w1} (\coqdocvar{con} \coqdocvar{c} \coqdocvar{b} \coqdocvar{w2}).\coqdoceol
\coqdocemptyline
\coqdocindent{1.00em}
\coqdockw{Definition} \coqdocvar{le} := \coqdocvar{le\_}.\coqdoceol
\coqdocemptyline
\coqdocindent{1.00em}
\coqdockw{Definition} \coqdocvar{le\_refl} : \coqdockw{\ensuremath{\forall}} \{\coqdocvar{w}\}, \coqdocvar{le} \coqdocvar{w} \coqdocvar{w}.\coqdoceol
\coqdocemptyline
\coqdocindent{1.00em}
\coqdockw{Definition} \coqdocvar{pforces} := \coqdockw{fun} \coqdocvar{w} \coqdocvar{p} \ensuremath{\Rightarrow} \coqdocvar{HSb} \coqdocvar{w} (\coqdocvar{prp} \coqdocvar{p}).\coqdoceol
\coqdocemptyline
\coqdocindent{1.00em}
\coqdockw{Definition} \coqdocvar{Answer} := \coqdocvar{Base}.\coqdoceol
\coqdocemptyline
\coqdocindent{1.00em}
\coqdockw{Definition} \coqdocvar{X} 
:= \coqdockw{fun} \coqdocvar{w} \coqdocvar{b} \ensuremath{\Rightarrow} \coqdocvar{HSb} \coqdocvar{w} \coqdocvar{b}.\coqdoceol
\coqdocnoindent
\coqdockw{End} \coqdocvar{structureHS}.\coqdoceol
\end{coqdoccode}


\end{framed}

Using the instantiated forcing structures, we can provide reification
functions for terms of the lambda calculus,
\begin{framed}
\begin{coqdoccode}
\coqdocnoindent
\coqdockw{Theorem} \coqdocvar{sreify} : (\coqdockw{\ensuremath{\forall}} \coqdocvar{F} \coqdocvar{w}, \coqdocvar{Cont} \coqdocvar{sforces} \coqdocvar{w} \coqdocvar{F} \ensuremath{\rightarrow} \coqdocvar{ND} \coqdocvar{w} \coqdocvar{F})\coqdoceol
\coqdocnoindent
\coqdockw{with} \coqdocvar{sreflect} : (\coqdockw{\ensuremath{\forall}} \coqdocvar{F} \coqdocvar{w}, \coqdocvar{ND} \coqdocvar{w} \coqdocvar{F} \ensuremath{\rightarrow} \coqdocvar{Cont} \coqdocvar{sforces} \coqdocvar{w} \coqdocvar{F}).\coqdoceol
\end{coqdoccode}


\end{framed}
\noindent
and for our compact terms:
\begin{framed}
\begin{coqdoccode}
\coqdocnoindent
\coqdockw{Theorem} \coqdocvar{creify} : \coqdoceol
\coqdocindent{1.00em}
(\coqdockw{\ensuremath{\forall}} \coqdocvar{c} \coqdocvar{w}, \coqdocvar{Cont} \coqdocvar{cforces} \coqdocvar{w} \coqdocvar{c} \ensuremath{\rightarrow} \coqdocvar{HSc} (\coqdocvar{explogn} \coqdocvar{c} (\coqdocvar{cnf} \coqdocvar{w})))\coqdoceol
\coqdocnoindent
\coqdockw{with} \coqdocvar{creflect} : (\coqdockw{\ensuremath{\forall}} \coqdocvar{c} \coqdocvar{w}, \coqdocvar{Cont} \coqdocvar{cforces} (\coqdocvar{ntimes} \coqdocvar{c} \coqdocvar{w}) \coqdocvar{c})\coqdoceol
\coqdocnoindent
\coqdockw{with} \coqdocvar{dreify} : (\coqdockw{\ensuremath{\forall}} \coqdocvar{d} \coqdocvar{w}, \coqdocvar{Cont} \coqdocvar{dforces} \coqdocvar{w} \coqdocvar{d} \ensuremath{\rightarrow} \coqdocvar{HSb} \coqdocvar{w} (\coqdocvar{bd} \coqdocvar{d}))\coqdoceol
\coqdocnoindent
\coqdockw{with} \coqdocvar{dreflect} : (\coqdockw{\ensuremath{\forall}} \coqdocvar{d} \coqdocvar{c1} \coqdocvar{c2} \coqdocvar{c3}, \coqdoceol
\coqdocindent{0.50em}
\coqdocvar{HSc} (\coqdocvar{explogn} \coqdocvar{c1} (\coqdocvar{cnf} (\coqdocvar{ntimes} \coqdocvar{c3} (\coqdocvar{con} \coqdocvar{c1} (\coqdocvar{bd} \coqdocvar{d}) \coqdocvar{c2})))) \ensuremath{\rightarrow}\coqdoceol
\coqdocindent{0.50em}
\coqdocvar{Cont} \coqdocvar{dforces} (\coqdocvar{ntimes} \coqdocvar{c3} (\coqdocvar{con} \coqdocvar{c1} (\coqdocvar{bd} \coqdocvar{d}) \coqdocvar{c2})) \coqdocvar{d}).\coqdoceol
\end{coqdoccode}


\end{framed}
\noindent
The reifier for atomic types, \textit{preify}, is not listed above,
because it is simply the `run' operation on the continuation monad. As
usually in TDPE, every reification function required its own
simultaneously defined reflection function.

Finally, one can combine the reifiers, the evaluators, and the
functions \textit{f2f} and \textit{f2f'}, in order to obtain both a
normalizing converter of lambda terms into compact terms (called
\textit{nbe} in the Coq implementation), and a converter
of compact terms into lambda terms (called \textit{ebn} in the Coq
implementation). One can, if one desires, also define only a partial
evaluator of lambda terms and only a partial evaluator of compact
terms.


\section{Conclusion}
\label{sec:conclusion}

\paragraph{Summary of our results}

We have brought into relation two distinct fundamental problems of the
lambda calculi underlying modern functional programming languages, one
concerning identity of types, and the other concerning identity of
terms, and we have shown how improved understanding of the first
problem can lead to improved understanding of the second problem.

We started by presenting a normal form of types, the exp-log normal
form, that is a systematic ordering of the high-school identities
allowing for a type to be mapped to normal form. This can be used as a
simple heuristic for deciding type isomorphism, a first such result
for the type language $\{\to,+,\times\}$. We beleive that the link
established to analysis and abstract algebra (the exp-log
decomposition produces a pair of homomorphisms between the additive
and the multiplicative group in an exponential field) may also be
beneficial to programming languages theory in the future.

The typing restrictions imposed to lambda terms in exp-log normal form
allowed us to decompose the standard axioms for $\betaetaeq$ into a
proper and simpler subset of themselves, $\betaetaeqe$. As far as we
are aware, this simpler axiomatization has not been isolated
before. Even more pleasingly, the new axiomatization disentangles the
old one, in the sense that left-hand sides and right-hand sides of the
equality axioms can no longer overlap.

Finally, we ended by giving a compact calculus of terms that can be
used as a more canonical alternative to the lambda calculus when
modeling the core of functional programming languages: the new syntax
does not allow for the $\eta$-axioms of Figure~\ref{fig:syntax:enf}
even to be stated, with the exception of \ref{e:eta:sum} that is still
present, albeit with a restricted type. As our method exploits type
information, it is orthogonal to the existing approaches that rely on
term analysis (discussed below), and hence could be used in addition
to them; we hope that it may one day help with addressing the part of
$\eta$-equality that is still beyond decision procedures.  We also
implemented and described a prototype converter from/to standard
lambda terms.

In the future, we would like to derive declarative rules to describe
more explicitly the extent of the fragment of $\betaetaeq$ decided by
our heuristic, although implicitly that fragment is determined by the
reduction to ENF congruence classes explained in
Section~\ref{sec:equality}. It should be noted that in this respect
our heuristic is no less explicitly described, than the only other
published one~\cite{balat_dicosmo_fiore} (reviewed below).

\paragraph{Related work}
Dougherty and Subrahmanyam \cite{dougherty_subrahmanyam} show that the
equational theory of terms (morphisms) for almost bi-Cartesian closed
categories is complete with respect to the set theoretic
semantics. This presents a generalization of Friedman's completeness
theorem for simply typed lambda calculus without sums (Cartesian
closed categories) \cite{Friedman1975}.

Ghani \cite{ghani1995} proves $\beta\eta$-equality of terms of the
lambda calculus with sum types to be decidable, first proceeding by
rewriting and eta-expansion, and then checking equality up to
commuting conversions by interpreting terms as finite sets of
quasi-normal forms; no canonical normal forms are obtained.

When sums are absent, the existence of a confluent and strongly
normalizing rewrite system proves the existence of canonical normal
forms, and then decidability is a simple check of syntactic identity
of canonical forms. Nevertheless, even in the context when sums are
absent, one may be interested in getting term representations that are
canonical \emph{modulo} type isomorphism, as in the recent works of
D{\'\i}az-Caro, Dowek, and Mart\'inez L\'opez
~\cite{DiazcaroDowek15,DiazcaroMartinezlopezIFL15}.

Altenkirch, Dybjer, Hofmann, and
Scott~\cite{altenkirch_dybjer_hofmann_scott} give another proof of
decidability of $\beta\eta$-equality for the lambda calculus with
sums by carrying out a normalization-by-evaluation argument in
category theory. They provide a canonical interpretation of the syntax
in the category of sheaves for the Grothendieck topology over the
category of constrained environments, and they claim that one can
obtain an algorithm for a decision procedure by virtue of the whole
development being formalizable in extensional Martin-Löf type theory.

In the absence of \ref{eta:sum} \cite{dougherty1993}, or for the
restriction of \ref{eta:sum} to $N$ being a variable
\cite{dicosmo_kesner}, a confluent and strongly normalizing rewrite
system exists, hence canonicity of normal forms for such systems
follows.

In \cite{balat_dicosmo_fiore}, Balat, Di Cosmo, and Fiore, present a
notion of normal form which is a syntactic counterpart to the notion
of normal forms in sheaves of Altenkirch, Dybjer, Hofmann, and
Scott. However, the forms are not canonical, as there may be two
different syntactic normal forms corresponding to a single semantic
one. They also say they believe (without further analysis or proof)
that one can get canonical normal forms if one considers an ordering
of nested $\delta$-expressions.

The normal forms of Balat et al. are sophisticated and determining if
something is a normal form relies on comparing sub-terms up to a
congruence relation $\approx$ on terms; essentially, this congruence
allows to identify terms such as the ones of our
{Example~\ref{example:5} and Example~\ref{example:6}}. For determining
if something is a normal form, in addition to the standard separation
of neutral vs normal terms, one uses three additional criteria: (A)
$\delta$-expressions that appear under a lambda abstraction must only
case-analyze terms involving the abstracted variable; (B) no two terms
which are equal modulo $\approx$ can be case analyzed twice; in
particular, no term can be case analyzed twice; (C) no case analysis
can have the two branches which are equal modulo $\approx$. To enforce
condition A, particular powerful control operators, set/cupto, are
needed in the implementation, requiring a patch of the ocaml
toplevel. Using our compact terms instead of lambda terms should help
get rid of condition A (hence set/cupto), since as we showed keeping a
constructor for $\lambda$'s in the representation of normal forms is
not necessary. On the other hand, we could profit from implementing
checks such as B and C in our implementation; however, our goal was to
see how far we can get in a purely type-directed way without doing any
program analysis.

A final small remark about this line of works: in \cite{balat2009},
Balat used the word ``canonical'' to name his normal forms, but this
does not preserve the usual meaning of that word, as showed in the
previous article~\cite{balat_dicosmo_fiore}.

Lindley \cite{lindley2007} presents another proof of decidability of
$\beta\eta$-equality for the lambda calculus with sums, based on an
original decomposition of the \ref{eta:sum}-axiom into four axioms
involving evaluation contexts (the proof of this decomposition,
Proposition~1, is unfortunately only sketched); the proof of
decidability uses rewriting modulo the congruence relation $\approx$
of Balat et al.

Scherer~\cite{scherer2015} reinterprets Lindley's rewriting approach
to decidability in the setting of the structural proof theory of
maximal multi-focusing, where he brings it in relation to the
technique of preemptive rewriting~\cite{Chaudhuri2008}. Scherer seems
to derive canonicity of his normal forms for natural deduction from
Lindley's results, although the later does not seem to show canonical
forms are a result of his rewriting decidability result.

The idea to apply type isomorphism in order to capture equality of
terms has been used before \cite{ahmad_licata_harper_manuscript}, but
only implicitly. Namely, in the \emph{focusing} approach to sequent
calculi~\cite{LiangMiller2007}, one gets a more canonical
representation of terms (proofs) by grouping all so called
asynchronous proof rules into blocks called asynchronous
phases. However, while all asynchronous proof rules are special kinds
of type isomorphisms, not all possible type isomorphisms are accounted
for by the asynchronous blocks: sequent calculi apply asynchronous
proof rules superficially, by looking at the top-most connectives, but
normalizing sequents (formulas) to their exp-log normal form applies
proof rules deeply inside the proof tree. Our approach can thus also
been seen as moving focusing proof systems into the direction of so
called deep inference systems.





\bibliographystyle{abbrvnat}
\bibliography{explog}





\end{document}